\documentclass[aps,twocolumn,showpacs,showkeys,pra,citeautoscript,amsmath,amssymb,floatfix,superscriptaddress,nofootinbib,10pt]{revtex4-1}
\pdfoutput=1

\usepackage{color}
\usepackage[dvipsnames]{xcolor}
\usepackage[normalem]{ulem}  
\usepackage{amsthm}
\usepackage{epsfig}
\usepackage[caption=false]{subfig}
\usepackage{hyperref}
\usepackage{xspace}

\definecolor{mylinkcolor}{rgb}{0,0,0.4}%
\hypersetup{unicode=true, %
  bookmarksnumbered=false,bookmarksopen=false,bookmarksopenlevel=1,%
  breaklinks=true,pdfborder={0 0 0},colorlinks=true}%
\hypersetup{%
  anchorcolor=mylinkcolor,citecolor=mylinkcolor, %
  filecolor=mylinkcolor,linkcolor=mylinkcolor, %
  menucolor=mylinkcolor,runcolor=mylinkcolor, %
  urlcolor=mylinkcolor}%

\usepackage{microtype}

\makeatletter
\newcommand\bibalias[2]{%
  \@namedef{bibali@#1}{#2}%
}

\newtoks\biba@toks
\let\bibalias@oldcite\cite
\def\cite{%
  \@ifnextchar[{%
    \biba@cite@optarg%
  }{%
    \biba@cite{}%
  }%
}
\newcommand\biba@cite@optarg[2][]{%
  \biba@cite{[#1]}{#2}%
}
\newcommand\biba@cite[2]{%
  \biba@toks{\bibalias@oldcite#1}%
  \def\biba@comma{}%
  \def\biba@all{}%
  \@for\biba@one@:=#2\do{%
    \edef\biba@one{\expandafter\@firstofone\biba@one@\@empty}%
    \@ifundefined{bibali@\biba@one}{%
      \edef\biba@all{\biba@all\biba@comma\biba@one}%
    }{%
      \PackageInfo{bibalias}{%
        Replacing citation `\biba@one' with `\@nameuse{bibali@\biba@one}'
      }%
      \edef\biba@all{\biba@all\biba@comma\@nameuse{bibali@\biba@one}}%
    }%
    \def\biba@comma{,}%
  }%
  \immediate\write\@auxout{\noexpand\bgroup\noexpand\renewcommand\noexpand\citation[1]{}\noexpand\citation{#2}\noexpand\egroup}%
  \edef\biba@tmp{\the\biba@toks{\biba@all}}%
  \biba@tmp%
}
\makeatother

\bibalias{360221697}{thermolimit}
\bibalias{Aberg14}{aberg2014catalytic}
\bibalias{BartlettRS07}{BRS-refframe-review}
\bibalias{BrandaoHNOW14}{brandao2013second}
\bibalias{FundLimits2}{HO-limitations}
\bibalias{HO2011-qthermo}{HO-limitations}
\bibalias{Hayden-embezzling--2}{Hayden-embezzling}
\bibalias{HorodeckiHHH09}{horodecki_quantum_2009}
\bibalias{PuszW78}{pusz_passive_1978}
\bibalias{YungerHalpernR14}{halpern2014beyond}
\bibalias{aaberg2013truly}{aaberg-singleshot}
\bibalias{adler2001gss}{adler2001generalized}
\bibalias{dahlsten2011inadequacy}{workvalue}
\bibalias{faist2012quantitative}{qlandauer}
\bibalias{gour_measuring_2009}{gour_measuring_2009--2}
\bibalias{horodecki_are_2002--2}{horodecki_are_2002}
\bibalias{horodecki_locking_2005--2}{horodecki_locking_2005}
\bibalias{horodecki_reversible_2003--2}{horodecki_reversible_2003}
\bibalias{lenard_thermodynamical_1978}{lenard78}
\bibalias{lenard_thermodynamical_1978--2}{lenard78}
\bibalias{linden2010smallest}{linden2010small}
\bibalias{linden_reversibility_2005--2}{linden_reversibility_2005}
\bibalias{marcothesis}{Tomamichel-thesis}
\bibalias{pusz_passive_1978--2}{pusz_passive_1978}
\bibalias{scovil1959maser}{Scovil1959masers}
\bibalias{synak-radtke_asymptotic_2006--2}{synak-radtke_asymptotic_2006}
\bibalias{wilming2014weak--2}{wilming2014weak}
\bibalias{wyner75}{Wyner-wiretap}

\newtheoremstyle{fornatcomm}%
{\parskip}
{\parskip}
{%
  \normalfont
}
{}
{}
{:}
{.6em}
{%
  {\itshape\thmname{#1}\thmnumber{ #2}}{\thmnote{ (#3)}}%
}
\theoremstyle{fornatcomm}

\usepackage{newfloat}
\DeclareFloatingEnvironment[name={Supplementary Figure}]{suppfigure}

\DeclareFloatingEnvironment[name={Figure}]{mainfigure}

\def\batt{ {\rm W} }
\def\bath{ {\rm R} }
\def\sys{ {\rm S} }
\def\anc{ {\rm A} }
\def\cat{ {\rm X} }

\def\id{\mathbbm{1}}
\newcommand{\kB}{k_\mathrm{B}}

\def\???{{\color{red}???}}

\def\final{\rho_{\rm S}'}  

\def\GTOlong{Non-Abelian Thermal Operations}
\def\GTOlongSing{Non-Abelian Thermal Operation}
\def\GTO{NATO}
\def\GGS{NATS}
\def\GGSlong{Non-Abelian Thermal State}

\definecolor{orange}{rgb}{1.0,0.3,0}

\def\<{\langle}
\def\>{\rangle}

\newcommand{\be}{\begin{eqnarray} \begin{aligned}}
\newcommand{\ee}{\end{aligned} \end{eqnarray} }
\newcommand{\benn}{\begin{eqnarray*} \begin{aligned}}
\newcommand{\eenn}{\end{aligned} \end{eqnarray*} }

\newcommand{\ben}{\begin{eqnarray} \begin{aligned}}
\newcommand{\een}{\end{aligned} \end{eqnarray} }

\newcommand{\bc}{\begin{center}}
\newcommand{\ec}{\end{center}}

\newcommand{\Tr}{\mathop{\mathsf{Tr}}\nolimits}

\newcommand{\norm}[1]{\left\| #1\right \|}

\newcommand{\abs}[1]{\left|#1 \right|}				

\newcommand{\beq}{\begin{eqnarray} \begin{aligned}}
\newcommand{\eeq}{\end{aligned} \end{eqnarray} }
\newcommand{\bea}{\begin{array}}
\newcommand{\eea}{\end{array}}

\newcommand{\bee}{\begin{enumerate}}
\newcommand{\eee}{\end{enumerate}}
\newcommand{\bei}{\begin{itemize}}
\newcommand{\eei}{\end{itemize}}

\newtheorem{theorem}{Theorem}
\newtheorem{proposition}[theorem]{Proposition}

\newtheorem{lemma}[theorem]{Lemma}

\newtheorem{definition}[theorem]{Definition}

\usepackage{amsfonts}

\def\id{\mathbb{I}}

\def\01{\{0,1\}}

\newcommand{\ket}[1]{|#1\rangle}

\newcommand{\ketbra}[2]{|#1\rangle\langle#2|}

\def\<{\langle}
\def\>{\rangle}

\def\nats{\gamma_{\mathbf{v}}}
\def\gibbsS{\gamma_{\rm S}}  
\def\gibbsBatt{\gamma_{ \batt }}  
\def\gibbsSBatt{\gamma_{\sys{\batt}}}  
\newcommand\gibbsParam[1]{\gamma_{#1}}
\def\final{\rho_{\rm S}'}  
\def\initialBatt{\rho_{ \batt }}  
\def\finalBatt{{\rho'_{ \batt }}}  

\def\initial{\rho_\sys}

\newcommand{\qalfree}{{\hat F}}
\newcommand{\qalfreesimple}{{\tilde F}}

\newcommand{\sgn}{\operatorname{sgn}}

\def\workf{{\cal W}}

\newcommand{\newreptheorem}[2]{%
\newenvironment{rep#1}[1]{%
 \def\rep@title{#2 \ref{##1} (restatement)}%
 \begin{rep@theorem}}%
 {\end{rep@theorem}}}
\makeatother

\newreptheorem{thm}{Theorem}
\newreptheorem{lem}{Lemma}

\def\paragraph#1{%
  \smallskip%
  \par\noindent{\textbf{#1}}\quad
}

\usepackage{bbm}

\usepackage{nameref}
\usepackage{zref-xr}
\zxrsetup{toltxlabel}
\zexternaldocument*[supp-]{NatComms.0.9}
\zexternaldocument*[maintext-]{NatComms.0.9}

\begin{document}

%
%
\author{Nicole~Yunger~Halpern}
\email{nicoleyh@caltech.edu}
\affiliation{Institute for Quantum Information and Matter, Caltech, Pasadena, CA 91125, USA}
\author{Philippe~Faist}
\email{pfaist@phys.ethz.ch}
\affiliation{Institute for Theoretical Physics, ETH Z\"{u}rich, 8093 Z\"{u}rich, Switzerland}
\author{Jonathan~Oppenheim}
\affiliation{Department of Physics and Astronomy, University College London, Gower Street, London WC1E 6BT, United Kingdom}
\author{Andreas~Winter}
\email{andreas.winter@uab.cat}
\affiliation{ICREA \& F\'{i}sica Te\`{o}rica: Informaci\'{o} i Fen\`{o}mens Qu\`{a}ntics, Universitat Aut\`{o}noma de Barcelona, ES-08193 Bellaterra (Barcelona), Spain}

\title{Microcanonical and resource-theoretic derivations \protect\\ of the thermal state 
of a quantum system with noncommuting charges}

\date{22 April 2016}

\begin{abstract}


The grand canonical ensemble lies at the core
of quantum and classical statistical mechanics.
A small system thermalizes to this ensemble
while exchanging heat and particles with a bath.
A quantum system may exchange
quantities represented by operators that fail to commute.
Whether such a system thermalizes,
and what form the thermal state has,
are questions about truly quantum thermodynamics.
We investigate this thermal state
from three perspectives.
First, we introduce an approximate microcanonical ensemble.
If this ensemble characterizes the system-and-bath composite,
tracing out the bath yields the system's thermal state.
This state is expected to be the equilibrium point,
we argue, of typical dynamics.
Finally, we define a resource-theory model
for thermodynamic exchanges of noncommuting observables.
Complete passivity---the inability to extract work from equilibrium states---implies 
the thermal state's form, too.
Our work opens new avenues into equilibrium 
in the presence of quantum noncommutation.


\end{abstract}

\maketitle

%
%
%

Recently reignited interest in quantum thermodynamics has prompted information-theoretic approaches to fundamental questions.
have enjoyed particular interest.~\cite{gemmer2009quantum,Gogolin2015arXiv_review,Goold2015arXiv_review,Vinjanampathy2015arXiv_review}.
The role of entanglement, for example, has been clarified with canonical
typicality~\cite{goldstein2006canonical,gemmer200418,PopescuSW06,LindenPSW09}.
Equilibrium-like behaviors have been predicted~\cite{FermiPU55,KinoshitaWW06,Rigol07,PolkovnikovSSV11}
and experimentally observed in
integrable quantum gases~\cite{Langen15,LangenGS15}.

Thermodynamic resource theories offer a powerful tool for analyzing
fundamental properties of the thermodynamics of quantum systems.
Heat exchanges with a bath are modeled
with ``free states'' and ``free operations''~\cite{janzing_thermodynamic_2000,BrandaoHORS13,brandao2013second,HO-limitations}.
These resource theories have been extended to model exchanges of additional physical
quantities, such as particles and angular
momentum~\cite{FundLimits2,VaccaroB11,YungerHalpernR14,YungerHalpern14,Weilenmann2015arXiv_axiomatic}.

A central concept in thermodynamics and statistical mechanics is the thermal state.
The thermal state has several important properties.
First, typical dynamics evolve the system toward the thermal state.
The thermal state is the equilibrium state. 
Second, consider casting statistical mechanics as an inference problem. 
The thermal state is the state which maximizes the entropy 
under constraints on physical quantities~\cite{Jaynes57I,Jaynes57II}.  
Third, consider the system as interacting with a large bath.
The system-and-bath composite occupies a microcanonical state.
Physical observables of the composite,
such as the total energy and total particle number,
have sharply defined values.
The system's reduced state is the thermal state.  
Finally, in a resource theory, 
the thermal state is the only completely passive state.
No work can be extracted from any number of copies of the thermal state~\cite{PuszW78,Lenard78}.

If a small system exchanges heat and particles with a large environment,
the system's thermal state
is a grand canonical ensemble:
$e^{ - \beta( H - \mu N ) } / Z$.
The system's Hamiltonian and particle number are represented by $H$ and $N$.
$\beta$ and $\mu$ denote the environment's 
inverse temperature and chemical potential.
The partition function $Z$ normalizes the state.
The system-and-bath dynamics conserves 
the total energy and total particle number.
More generally, subsystems exchange conserved quantities, or  ``charges,''
$Q_j,  \;  \:  j = 1, 2, \ldots c$.  To these charges correspond
generalized chemical potentials $\mu_j$.
The $\mu_j$'s characterize the bath.

We address the following question.
Suppose that the charges fail to commute with each other:
$[Q_j,  Q_k]  \neq  0$.
What form does the thermal state have?
We call this state ``the \GGSlong{}'' (\GGS{}).
Jaynes applied the Principle of Maximum Entropy to this question~\cite{Jaynes57II}.
He associated fixed values $v_j$ with the charges' expectation values.
He calculated the state that,
upon satisfying these constraints,
maximizes an entropy.
This thermal state has a generalized Gibbs form:
\begin{align}
   \label{eq:GGS}
   \gamma_{ \mathbf{v} }
   :=   \frac{1}{Z}   e^{ - \sum_{j = 0}^c  \mu_j  Q_j } \ ,
\end{align}
wherein the  the $v_j$'s determine the $\mu_j$'s.

Our contribution is a mathematical, physically justified derivation of the thermal state's form
for systems whose dynamics conserve noncommuting observables.
We recover the state~\eqref{eq:GGS} 
via several approaches, demonstrating its physical importance.
We address  puzzles raised in~\cite{YungerHalpern14,Imperial15} about how
to formulate a resource theory in which thermodynamic charges fail to commute.
Closely related, independent work was performed by Guryanova \emph{et al.}~\cite{teambristol}.
We focus primarily on the nature of passive states.
Guryanova \emph{et al.}, meanwhile, focus more 
on the resource theory for multiple charges
and on tradeoffs amongst types of charge extractions.

In this paper, we derive the \GGS{}'s form from a microcanonical argument.
A simultaneous eigenspace of all the noncommuting
physical charges might not exist.
Hence we introduce the notion of an approximate microcanonical subspace.
This subspace consists of the states in which
the charges have sharply defined values.
We derive conditions under which this subspace exists.
We show that a small subsystem's reduced state
lies, on average, close to $\gamma_{ \mathbf{v} }$.
Second, we invoke canonical typicality~\cite{PopescuSW06,LindenPSW09}.
If the system-and-bath composite 
occupies a random state in the approximate microcanonical subspace, 
we argue, a small subsystem's state likely lies close to the \GGS{}. 
Typical dynamics are therefore expected to evolve 
a well-behaved system's state towards the \GGS{}.
Third, we define a resource theory for thermodynamic exchanges 
of noncommuting conserved charges.
We extend existing resource theories 
to model the exchange of noncommuting quantities. 
We show that the \GGS{} is the only possible free state 
that renders the theory nontrivial: 
 Work cannot be extracted from 
any number of copies of $\gamma_{ \mathbf{v} }$.  
We show also that the \GGS{} is the only state preserved by free operations.  
From this preservation, we derive ``second laws'' that govern state transformations.
This work provides a well-rounded, and novelly physical, perspective
on equilibrium in the presence of quantum noncommutation.
This perspective opens truly quantum avenues in thermodynamics.

%
%
%
\section{Results}
\label{section:Results}

\subsection{Overview}

We derive the \GGSlong{}'s form via three routes:
from a microcanonical argument,
from a dynamical argument
built on canonical typicality,
and from complete passivity in a resource theory.
Details appear in Supplementary Notes~\ref{supp-section:SI_Micro}--\ref{supp-section:SI_RT}.

\subsection{Microcanonical derivation}

In statistical mechanics, the form $e^{- \beta ( H - \mu N) } / Z$
of the grand canonical ensemble is well-known to be derivable 
as follows.
The system of interest is assumed 
to be part of a larger system.
Observables of the composite have fixed values $v_j$.
For example, the energy equals $E_0$,
and the particle number equals $N_0$.
The microcanonical ensemble
is the whole-system state spread uniformly across 
these observables' simultaneous eigenspace.
Tracing out the environmental degrees of freedom 
yields the state $e^{ - \beta (H - \mu N) } / Z$.

We derive the \GGS{}'s form similarly.
Crucially, however, we adapt the above strategy 
to allow for noncommuting observables.
Observables might not have well-defined values $v_j$ simultaneously.
Hence a microcanonical ensemble as discussed above, suitable for commuting 
observables, may not exist.
We overcome this obstacle by introducing an
approximate microcanonical ensemble $\Omega$.
We show that, for every state satisfying the conditions of
an approximate microcanonical ensemble, tracing out most of the larger system
yields, on average, a state close to the \GGS{}.
We exhibit conditions under which an approximate 
microcanonical ensemble exists.
The conditions can be satisfied when the
larger system consists of many noninteracting replicas of the system.
An important step in the proof consists of reducing the 
noncommuting case to the commuting one.
This reduction relies on a result by Ogata~\cite[Theorem 1.1]{Ogata11}.
A summary appears in Fig.~\ref{fig:Overview}.

%
%
%
%
\begin{mainfigure}
  \centering
  \includegraphics[width=87mm]{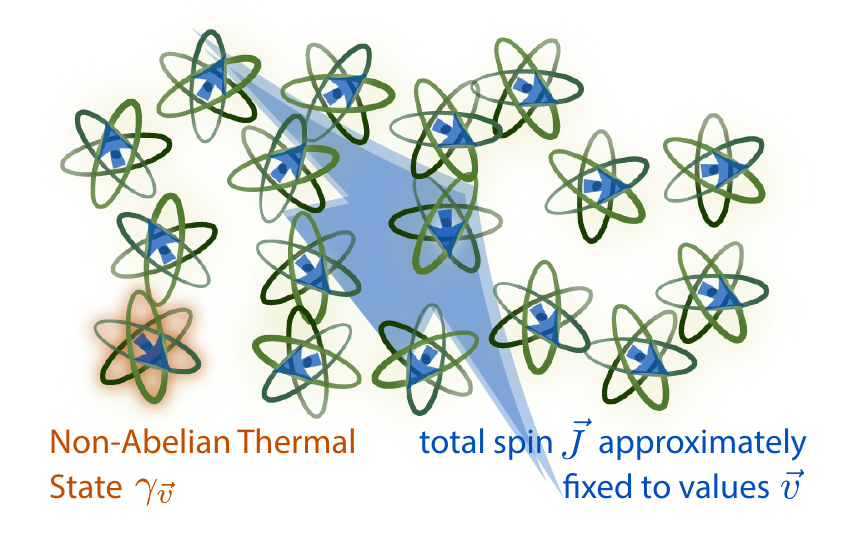}
  \caption{\textbf{\GGSlong{}: }
    We derive the form of the thermal state
of a system that has charges
that might not commute with each other.
Example charges include the components $J_i$
of the spin $\mathbf{J}$.
We derive the thermal stateÕs form
by introducing an approximate microcanonical state.
An ordinary microcanonical ensemble
could lead to the thermal stateÕs form 
if  the charges commuted:
Suppose, for example, that the charges were 
a Hamiltonian $H$
and a particle number $N$
that satisfied $[H, N] = 0$.
Consider many copies of the system.
The composite system could have a well-defined energy $E_{\rm tot}$
and particle number $N_{\rm tot}$ simultaneously.
$E_{\rm tot}$ and $N_{\rm tot}$ would correspond to 
some eigensubspace $\mathcal{H}_{E_{\rm tot}, N_{\rm tot}}$
shared by the total Hamiltonian
and the total-particle-number operator.
The (normalized) projector onto $\mathcal{H}_{E_{\rm tot}, N_{\rm tot}}$
would represent the composite systemÕs microcanonical state.
Tracing out the bath would yield the systemÕs thermal state.
But the charges $J_i$ under consideration
might not commute.
The charges might share no eigensubspace.
Quantum noncommutation demands a modification 
of the ordinary microcanonical argument.
We define an approximate microcanonical subspace $\mathcal{M}$.
Each state in $\mathcal{M}$ 
simultaneously has almost-well-defined values 
of noncommuting whole-system charges:
Measuring any such whole-system charge
has a high probability
of outputting a value
close to an ``expectedÕÕ value 
analogous to $E_{\rm tot}$ and $N_{\rm tot}$.
We derive conditions under which 
the approximate microcanonical subspace $\mathcal{M}$ exists.
The (normalized) projector onto $\mathcal{M}$
represents the whole systemÕs state.
Tracing out most of the composite system yields 
the reduced state of the system of interest.
We show that the reduced state is, on average,
close to the \GGSlong{} (\GGS{}).
This microcanonical derivation of the \GGS{}'s form
links Jaynes's information-theoretic derivation to physics.}
  \label{fig:Overview}
\end{mainfigure}

%
%
%
%
Set-up:  
Let $\mathcal{S}$ denote a system
associated with a Hilbert space $\mathcal{H}$;
with a Hamiltonian $H  \equiv  Q_0$;
and with observables (which we call ``charges'') $Q_1, Q_2, \ldots, Q_c$.
The charges do not necessarily commute with each other:
$[Q_j,  Q_k]  \neq  0$.

Consider $N$ replicas of $\mathcal{S}$, associated with 
the composite system Hilbert space $\mathcal{H}^{\otimes N}$.
We average each charge $Q_j$ over the $N$ copies:
\begin{align}
  \bar{Q}_j := \frac{1}{N}  \sum_{ \ell = 0}^{ N - 1}  
                \id^{\otimes \ell }  \otimes  Q_j  \otimes  \id^{ \otimes (N - 1 - \ell) }.
\end{align}
The basic idea is that, as $N$ grows, the averaged operators $\bar{Q}_j$
come increasingly to commute. 
Indeed, there exist operators operators $\bar{Y}_j$ 
that commute with each other
and that approximate the averages~\cite[Theorem 1.1]{Ogata11}.
An illustration appears in Fig.~\ref{fig:OgataProofSetup}.

\begin{mainfigure}
  \centering
  \includegraphics{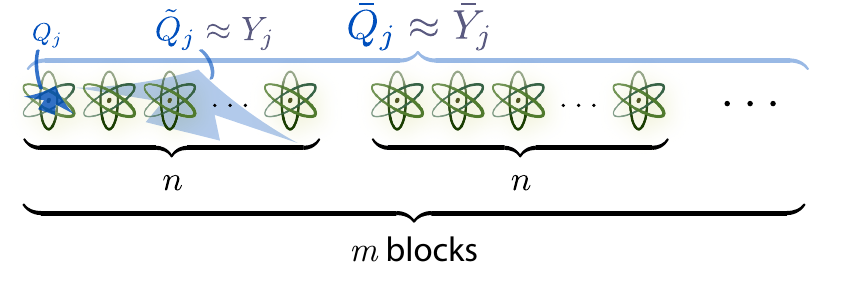}
  \caption{\textbf{Noncommuting charges:  }
We consider a thermodynamic system $\mathcal{S}$
that has conserved charges $Q_j$.
These $Q_j$'s might not commute with each other.
The system occupies a thermal state
whose form we derive.
The derivation involves an approximate microcanonical state
of a large system that contains the system of interest.
Consider a block of $n$ copies of $\mathcal{S}$.
Most copies act, jointly, similarly to a bath
for the copy of interest.
We define $\tilde Q_j$ as the average of the $Q_j$'s 
of the copies in the block. 
Applying results from Ogata~\protect\cite{Ogata11},
we find operators $\tilde{Y}_j$
that are close to the $\tilde Q_j$'s
and that commute with each other. 
Next, we consider $m$ such blocks.
This set of $m$ blocks
contains $N=mn$ copies of $\mathcal{S}$. 
Averaging the $\tilde{Q}_j$'s over the blocks, for a fixed $j$-value, 
yields a global observable $\bar Q_j$.
The $\bar Q_j$'s are approximated by $\bar Y_j$'s.
The $\bar Y_j$'s are the corresponding averages of the $\tilde{Y}_j$'s.
The approximate global charges $\bar{Y}_j$ 
commute with each other.
The commuting $\bar{Y}_j$'s
enable us to extend the concept of a microcanonical ensemble
from the well-known contexts in which all charges commute
to truly quantum systems whose charges do not necessarily commute.}
\label{fig:OgataProofSetup}
\end{mainfigure}

%
%
%
%
Derivation:
Since the $\bar{Y}_j$'s commute mutually, they can be measured
simultaneously. 
More importantly, the joint Hilbert space $\mathcal{H}^{\otimes n}$
contains a subspace
on which each $\bar{Q}_j$ has 
prescribed values close to $v_j$.
Let $\mathcal{M}$ denote the subspace.
Perhaps unsurprisingly, because the $\bar{Y}_j$'s approximate the $\bar{Q}_j$'s,  
each state in $\mathcal{M}$ 
has a nearly well-defined value of $\bar{Q}_j$ near $v_j$.
If $\bar{Q}_j$ is measured, the distribution
is sharply peaked around $v_j$. 
We can also show the opposite: every state with nearly well-defined values $v_j$ of all
$\bar{Q}_j$'s has most of its probability weight in $\mathcal{M}$.

These two properties show that $\mathcal{M}$ is
an approximate microcanonical subspace for the $\bar{Q}_j$'s
with values $v_j$. The notion of the approximate microcanonical subspace
is the first major contribution of our work. It captures the idea 
that, for large $N$, we can approximately fix the values of the
noncommuting charges $Q_j$.
An approximate microcanonical subspace $\mathcal{M}$ is any subspace
consisting of the whole-system states whose average observables $\bar{Q}_j$
have nearly well-defined values $v_j$. More precisely,
a measurement of any $\bar{Q}_j$ has a high probability of yielding 
a value near $v_j$ if and only if most of the state's probability weight lies 
in $\mathcal{M}$.

%
%
Normalizing the projector onto $\mathcal{M}$
yields an approximate microcanonical ensemble, $\Omega$.
Tracing out every copy of $\mathcal{S}$ but the $\ell^{\rm th}$
yields the reduced state $\Omega_\ell$.
The distance between $\Omega_\ell$
and the \GGS{} $\gamma_{ \mathbf{v} }$
can be quantified by the relative entropy
\begin{align} 
   \label{eq:RelEntMain}
   D ( \Omega_\ell \| \gamma_{ \mathbf{v} } )
   :=  - S( \Omega_\ell )
        -  \Tr (  \Omega_\ell    \log \gamma_{ \mathbf{v} } ).
\end{align}
Here, $S ( \Omega_\ell )  :=  - \Tr ( \Omega_\ell  \log  \Omega_\ell )$
is the von Neumann entropy.
The relative entropy $D$ is bounded by the trace norm $\| . \|_1$,
which quantifies the distinguishability 
of $\Omega_\ell$ and $\gamma_{ \mathbf{v} }$~\cite{HiaiOT81}:
\begin{align}
  \label{eq:PinskerIneq}
  D(\Omega_\ell\Vert\gamma_{\mathbf v})
  \geq \frac12 \norm{\Omega_\ell-\gamma_{\mathbf v}}_1^2.
\end{align}

Our second main result is that, if $\Omega$ is an approximate
microcanonical ensemble, then the average, over systems $\ell$,
of the relative entropy $D$ between $\Omega_\ell$ and $\gamma_{ \mathbf{v} }$
is small:
\begin{align}
     \label{eq:AvgRelEntMain}
      \frac{1}{N} \sum_{ \ell = 0}^{ N - 1} D ( \Omega_\ell \| \gamma_{ \mathbf{v} } )
       \leq \theta + \theta'.
\end{align}
The parameter $\theta  = \mathrm{(const.)} / \sqrt{N}$
vanishes in the many-copy limit.
$\theta'$ depends on the number $c$ of charges,
on the approximate expectation values $v_j$,
on the eigenvalues of the charges $Q_j$,
and on the (small) parameters in terms of which 
$\mathcal{M}$ approximates a microcanonical subspace.

Inequality~\eqref{eq:AvgRelEntMain} capstones the derivation.
The inequality follows from bounding each term in Eq.~\eqref{eq:RelEntMain},
the definition of the relative entropy $D$.
The entropy $S( \Omega_\ell )$ is bounded with $\theta$.
This bound relies on Schumacher's Theorem,
which quantifies the size of a high-probability subspace like $\mathcal{M}$
with an entropy $S(\gamma_{ \mathbf{v} })$~\cite{Schumacher95}.
We bound the second term in the $D$ definition with $\theta'$.
This bound relies on the definition of $\mathcal{M}$:
Outcomes of measurements of the $\bar{Q}_j$'s
are predictable up to parameters
on which $\theta'$ depends.

%
%
Finally, we present conditions
under which the approximate microcanonical subspace $\mathcal{M}$ exists.
Several parameters quantify the approximation.
The parameters are shown to be interrelated
and to approach zero simultaneously as $N$ grows. In particular,
the approximate microcanonical subspace 
$\mathcal{M}$ exists if $N$ is great enough.

%
%
%
%
This microcanonical derivation offers a physical counterpoint
to Jaynes's maximum-entropy derivation of the \GGS{}'s form.
We relate the \GGS{} to the physical picture
of a small subsystem in a vast universe
that occupies an approximate microcanonical state.
This vast universe allows the Correspondence Principle
to underpin our argument.
In the many-copy limit as $N \to \infty$,
the principle implies that quantum behaviors should vanish,
as the averages of the noncommuting charges $Q_j$
come to be approximated by commuting $\bar{Y}_j$'s.
Drawing on Ogata's~\cite[Theorem 1.1]{Ogata11},
we link thermality in the presence of noncommutation 
to the physical Correspondence Principle.


\subsection{Dynamical considerations}

The microcanonical and maximum-entropy arguments
rely on kinematics and information theory.
But we wish to associate the \GGS{}
with the fixed point of dynamics.
The microcanonical argument, combined with canonical typicality,
suggests that the \GGS{} is the equilibrium state of typical dynamics.
Canonical typicality enables us to model
the universe's state with a pure state
in the approximate microcanonical subspace $\mathcal{M}$.
If a large system occupies a randomly chosen pure state,
the reduced state of a small subsystem
is close to thermal~\cite{goldstein2006canonical,gemmer200418,PopescuSW06,LindenPSW09}.

Consider, as in the previous section, $N$ copies of the system $\mathcal{S}$.
By $\Omega$, we denoted
the composite system's approximately microcanonical state.
We denoted by $\Omega_\ell$
the reduced state of the $\ell^{\rm th}$ copy,
formed by tracing out most copies from $\Omega$.
Imagine that the whole system 
occupies a pure state $\ket{ \psi }  \in  \mathcal{M}$. 
Denote by $\rho_\ell$
the reduced state of the $\ell^{\rm th}$ copy.
$\rho_\ell$ is close to $\Omega_\ell$, on average,
by canonical typicality~\cite{PopescuSW06}:
\begin{equation}
   \<  \|  \rho_\ell -  \Omega_\ell  \|_1 \> \leq  \frac{d}{\sqrt{D_M}}.
   \label{eq:typical}
\end{equation}
The average $\langle . \rangle$ is over pure states $\ket{ \psi } \in \mathcal{M}$.
The trace norm is denoted by $ \| . \|_1$;
$d  :=  {\rm dim} ( \mathcal{H} )$ denotes the dimensionality
of the Hilbert space $\mathcal{H}$ of one copy of $\mathcal{S}$;
and $D_M  :=  {\rm dim} ( \mathcal{M} )$ denotes the dimensionality
of the approximate microcanonical subspace $\mathcal{M}$.

We have bounded, using canonical typicality, 
the average trace norm between $\rho_\ell$ and $\Omega_\ell$.
We can bound the average trace norm
between $\Omega_\ell$ and the \GGS{} $\gamma_{ \mathbf{v} }$,
using our microcanonical argument.
[Supplementary equation~\eqref{supp-eq:RelEnt} bounds the average relative entropy $D$
between $\Omega_\ell$ and  $\gamma_{ \mathbf{v} }$.
Pinsker's Inequality, Ineq.~\eqref{eq:PinskerIneq},
lower bounds $D$ in terms of the trace norm.]
Combining these two trace-norm bounds via the Triangle Inequality,
we bound the average distance between $\rho_\ell$ and $\gamma_{ \mathbf{v} }$:
\begin{equation}
   \biggl\langle    \frac1N    \sum_{\ell=0}^{N-1} 
         \| \rho_\ell   -    \gamma_{\mathbf{v}}  \|_1    \biggr\rangle
   \leq    \frac{d}{ \sqrt{D} }
          +   \sqrt{ 2 (\theta + \theta') }.
\label{eq:typical-2}
\end{equation}
If the whole system occupies a random pure state $\ket{ \psi }$
in $\mathcal{M}$,
the reduced state $\rho_\ell$ of a subsystem 
is, on average, close to the \GGS{} $\gamma_{ \mathbf{v} }$.

Sufficiently ergodic dynamics
is expected to evolve the whole-system state
to a $\ket{\psi}$ that satisfies Ineq.~\eqref{eq:typical-2}:
Suppose that the whole system begins in a pure state $\ket{ \psi (t{=}0) } \in \mathcal{M}$.
Suppose that the system's Hamiltonian
commutes with the charges: $[H,Q_j]=0$ for all $j=1,\ldots,c$.
The dynamics conserves the charges.
Hence most of the amplitude of $\ket{ \psi (t) }$
remains in $\mathcal{M}$ for appreciable times.
Over sufficient times, ergodic dynamics yields 
a state $\ket{ \psi(t) }$ that can be regarded as random.
Hence the reduced state is expected be close to 
$\Omega_\ell \approx \gamma_{ \mathbf{v}}$ for most long-enough times $t$.

Exploring how the dynamics depends on the number of copies of the system offers promise for interesting future research.



\subsection{Resource theory}

A thermodynamic resource theory is an explicit characterization of 
a thermodynamic system's resources, free states, and free operations 
with a rigorous mathematical model.
The resource theory specifies what 
an experimenter considers valuable (e.g., work) and what
is considered plentiful, or free (e.g., thermal states). 
To define a resource theory, we specify allowed operations and
which states can be accessed for free.
We use this framework to quantify the resources needed to transform one state into another. 

The first resource theory was entanglement theory~\cite{HorodeckiHHH09}.
The theory's free operations are local operations and classical communication (LOCC).
The free states are the states which can be easily prepared with LOCC,
the separable states. Entangled states constitute valuable resources. 
One can quantify entanglement using this resource theory. 

We present a resource theory for thermodynamic systems 
that have noncommuting conserved charges $Q_j$.
The theory is defined by its set of free operations, which we call ``{\GTOlong}'' (\GTO).
\GTO{} generalize thermal operations~\cite{janzing_thermodynamic_2000,FundLimits2}.  
How to extend thermodynamic resource theories 
to conserved quantities other than energy was noted in~\cite{FundLimits2,YungerHalpernR14,YungerHalpern14}. 
The {\GTO} theory is related to the resource theory in \cite{Imperial15}.

We supplement these earlier approaches with two additions.
First, a battery has a work payoff function dependent on chemical potentials.
We use this payoff function to define chemical work.
Second, we consider a reference system for a non-Abelian group.
The reference system is needed to resolve the difficulty encountered in \cite{YungerHalpern14,Imperial15}: There
might be no nontrivial operations which respect all the conservation laws. 
The laws of physics require that 
any operation performed by an experimenter
commutes with all the charges. 
If the charges fail to commute with each other, there might be no nontrivial unitaries
which commute with all of them. 
In practice, one is not limited by such a stringent constraint. 
The reason is that an experimenter has access to a reference frame~\cite{aharonov-susskind,kitaev2014super,BRS-refframe-review}.  
  
A reference frame is a system $W$
prepared in a state such that, for any unitary on a system $S$ which does not
commute with the charges of $S$, some global unitary on $WS$ conserves the total
charges and approximates the unitary on $S$ to arbitrary precision.  
The reference frame relaxes the strong constraint on the unitaries.  
The reference frame can be merged with the battery, 
in which the agent stores the ability to perform work.
We refer to the composite as ``the battery.''
We denote its state by $\rho_\batt$. 
The battery has a Hamiltonian $H_\batt$ and charges $Q_{j_\batt}$, described below.

Within this resource theory, the \GGSlong{} emerges in two ways:
\begin{enumerate}
   \item 
   The \GGS{} is the unique state from which work cannot be extracted,
   even if arbitrarily many copies are available. 
   That is, the \GGS{} is completely passive.
   \item The \GGS{} is the only state of $S$ that remains invariant under the free operations during which no work is performed on $S$.
\end{enumerate}

Upon proving the latter condition, we prove second laws for
thermodynamics with noncommuting charges.
These laws provide necessary conditions for a transition to be possible.
In some cases, we show, the laws are sufficient. 
These second laws govern state transitions of 
a system $\rho_\sys$, governed by a  Hamiltonian $H_\sys$,
whose charges $Q_{j_\sys}$ can be exchanged with the surroundings.
We allow the experimenter to couple  $\rho_\sys$ to free states $\rho_\bath$.
The form of $\rho_\bath$ is determined by the Hamiltonian $H_\bath$ and the charges $Q_{j_\bath}$
attributable to the free system.
We will show that these free states have the form of the \GGS. 
As noted above, no other state could be free.
If other states were free, an arbitrarily large amount of work 
could be extracted from them.

Before presenting the second laws,  
we must define ``work.''  
In textbook examples about gases, one defines
work as $\delta W = p\,dV$, because a change in volume at a fixed pressure can be
translated into the ordinary notion of mechanical work.  
If a polymer is stretched, then
$\delta W = F\,dx$, wherein $x$ denotes the polymer's linear displacement
and $F$ denotes the restoring force.
If $B$ denotes a magnetic field and $M$ denotes a medium's magnetization,
$\delta W = B\,dM$.  
The definition of ``work'' can depend on one's
ability to transform changes in thermodynamic variables into a standard
notion of ``work,'' such as mechanical or electrical work.

Our approach is to define a notion of chemical work. 
We could do so by modelling explicitly how the change in some quantity $Q_j$ 
can be used to extract $\mu_j \, \delta Q_j$ work.
Explicit modelling would involve adding a term to the battery Hamiltonian $H_\batt$. 
Rather than considering a specific work Hamiltonian or model of chemical work, 
however, we consider a work payoff function,
\begin{align}
  \workf= \sum_{j=0}^c\mu_j Q_{j_\batt}\ .
  \label{eq:workfunc}
\end{align}
The physical situation could determine the form of this $\workf$.
For example, the $\mu_j$'s could denote the battery's chemical potentials.
In such a  case, $\workf$ would denote the battery's total Hamiltonian, 
which would depend on those potentials. 

We choose a route conceptually simpler than considering 
an explicit Hamiltonian and battery system, however.
We consider Eq.~\eqref{eq:workfunc} as a payoff function that 
defines the linear combination of charges that interests us.
We define the (chemical) work expended or distilled during a transformation
as the change
in the quantum expectation value $\langle \workf \rangle$.

The form of $\workf$ is implicitly determined by the battery's structure
and by how charges can be converted into work.
For our purposes, however, the origin of the form of $\workf$ need not be known.
$\workf$ will uniquely determine the $\mu_j$'s in the \GGS.
Alternatively, we could first imagine that the agent could access, 
for free, a particular \GGS.
This \GGS{}'s form would determine the work function's form. 
If the charges commute, the corresponding Gibbs state is known to be 
the unique state that is completely passive with respect to the  observable~\eqref{eq:workfunc}.


%

In Supplementary Note~\ref{supp-section:SI_RT},
we specify the resource theory for noncommuting charges in more detail.
We show how to construct allowable operations, using the reference frame and battery. 
From the allowable operations, we derive a zeroth law of thermodynamics.
 
Complete passivity and zeroth law:
This zeroth law relates to the principle of complete passivity,
discussed in~\cite{PuszW78,Lenard78}.
A state is complete passive if,
an agent cannot extract work
from arbitrarily many copies of the state.
In the resource theory for heat exchanges, 
completely passive states can be free. 
They do not render the theory trivial 
because no work can be drawn from them~\cite{brandao2013second}.

In the \GTO{} resource theory, we show, the only reasonable free states 
have the \GGS{}'s form.
The free states' chemical potentials equal 
the $\mu_j$'s in the payoff function $\workf$, 
at some common fixed temperature.
 Any other state would render the resource theory trivial: 
 From copies of any other state,
arbitrary much work could be extracted for free. 
Then, we show that the \GGS{} is preserved by \GTO{},
the operations that perform no work on the system.
  
The free states form an equivalence class.
They lead to notions of temperature and chemical potentials $\mu_j$.
This derivation of the free state's form extends
complete passivity and the zeroth law from~\cite{brandao2013second} 
to noncommuting conserved charges.
The derivation further solidifies the role of the \GGSlong{} in thermodynamics.

%
%
%
%
Second laws:
The free operations preserve the \GGS{}.
We therefore focus on contractive measures 
of states' distances from the \GGS{}.
Contractive functions decrease monotonically under the free operations. 
Monotones feature in ``second laws'' that signal whether 
\GTO{} can implement a state transformation.
For example, the $\alpha$-R\'enyi relative entropies 
between a state and the \GGS{} cannot increase.  

Monotonicity allows us to define generalized free energies as
\begin{align}
  F_\alpha   \left(\rho_\sys,      \gibbsS  \right) 
  := \kB T D_\alpha\left(\rho_\sys   \Vert\gibbsS\right) 
  - \kB T\log Z\ ,
\end{align}
wherein $\beta  \equiv  1/ (k_{\rm B} T)$ and $\kB$ denotes Boltzmann's constant.
$\gibbsS$ denotes the \GGS\ with respect to 
the Hamiltonian $H_\sys$ and the charges $Q_{j_\sys}$ of the system $S$.
The partition function is denoted by $Z$.
Various classical and quantum definitions of the R\'enyi relative entropies $D_\alpha$ 
are known to be contractive~\cite{brandao2013second,HiaiMPB2010-f-divergences,Muller-LennertDSFT2013-Renyi,WildeWY2013-strong-converse,JaksicOPP2012-entropy}.
The free energies $F_\alpha$ decrease monotonically
if no work is performed on the system.
Hence the $F_\alpha$'s characterize 
natural second laws that govern achievable transitions.

For example, the classical R{\'e}nyi divergences $D_\alpha(\initial\|\gibbsS)$ 
are defined as 
\begin{equation}
   D_\alpha(\initial   \|   \gibbsS)
   := \frac{\sgn(\alpha)}{\alpha-1} \log 
   \left(   \sum_k    p_k^\alpha    q_k^{1-\alpha}   \right),
   \label{eq:renyidivergence}
\end{equation}
wherein $p_k$ and $q_k$ denote the probabilities 
of $\initial$ and of $\gibbsS$ in the $\workf$ basis. 
The $D_\alpha$'s lead to second laws that hold 
even in the absence of a reference frame
and even outside the context of the average work.


The $F_\alpha$'s reduce to the standard free energy
when averages are taken over large numbers.
Consider the asymptotic (``thermodynamic'') limit 
in which many copies $( \initial )^{\otimes n}$ of $\initial$ are transformed.
Suppose that the agent has some arbitrarily small probability $\varepsilon$ of failing 
to implement the desired transition.
$\varepsilon$ can be incorporated into the free energies via a technique called ``smoothing''~\cite{brandao2013second}.
The average, over copies of the state, of every smoothed 
$F^\varepsilon_\alpha$ approaches $F_1$~\cite{brandao2013second}:
\begin{align}
  \lim_{n\rightarrow\infty}    \frac{1}{n}   
  F^\varepsilon_\alpha  & \Big(  
       (\initial)^{\otimes n},   ( \gibbsS )^{\otimes n}   \Big)
  =  F_1  \\
  & =   k_B T D(\rho_\sys   \|\gibbsS) -k_BT\log ( Z )   \\
  & =    
       \< H_\sys  \>_{\initial}   -   k_B TS(\initial)   
       +   \sum_{j = 1}^c   \mu_j   \<Q_{j_\sys}\>.
  \label{eq:freeenergy}
\end{align}
We have invoked the relative entropy's definition,
\begin{align}
   D(\rho_\sys   \|   \gibbsS)  :=   \Tr     \Big(  \rho_\sys   \log ( \rho_\sys ) \Big) 
   -  \Tr     \Big(  \rho_\sys   \log (\gibbsS)  \Big).
\end{align}
 Note the similarity between the many-copy average $F_1$ 
  in Eq.~\eqref{eq:freeenergy}
and the ordinary free energy,
$F = E - T \, dS + \sum_j \mu_j \, dN_j$.
The monotonic decrease of $F_1$ 
constitutes a necessary and sufficient condition 
for a state transition to be possible in the presence of a reference system
in the asymptotic limit.

In terms of the generalized free energies, we formulate second laws:
\begin{proposition}
  \label{prop:second-laws}
  In the presence of a heat bath 
  of inverse temperature $\beta$ and chemical potentials $\mu_j$, 
  the free energies   $F_\alpha(\initial,   \gibbsS)$ 
  decrease monotonically:
  \begin{align}
     F_{\alpha}(\initial,   \gibbsS) \geq F_{\alpha}(\final,  \gibbsS')
     \; \: \forall \alpha\geq 0,
  \end{align}
  wherein $\initial$ and $\final$ denote the system's initial and final states.
  The system's Hamiltonian and charges may transform
  from $H_\sys$ and $Q_{j_\sys}$ to $H'_\sys$ and $Q_{j_\sys}'$.
  The \GGS{}s associated with the same Hamiltonians and charges
  are denoted by $\gibbsS$ and $\gibbsS'$.
If
\begin{align}
& [\mathcal{W},  \final ]  =  0
    \quad {\rm and} \quad \nonumber\\
&   F_{\alpha}(\initial,  \gibbsS) \geq F_{\alpha}(\final, \gibbsS')
\;   \;   \forall   \alpha \geq 0,
\end{align}
some \GTO\ maps $\initial$ to $\final$.
\end{proposition}

As in~\cite{brandao2013second}, additional laws can be defined 
in terms of quantum R{\'e}nyi divergences~\cite{HiaiMPB2010-f-divergences,Muller-LennertDSFT2013-Renyi,WildeWY2013-strong-converse,JaksicOPP2012-entropy}. 
This amounts to choosing, in Proposition~\ref{prop:second-laws}, 
a definition of the R\'enyi divergence which accounts for 
the possibility that $\initial$ and $\final$ 
have coherences relative to the $\workf_\sys$ eigenbasis.
Several measures are known to be
contractive~\cite{HiaiMPB2010-f-divergences,Muller-LennertDSFT2013-Renyi,WildeWY2013-strong-converse,JaksicOPP2012-entropy}.
They, too, provide a new set of second laws. 


%
%
%
%
Extractable work:
In terms of the free energies $F_\alpha$, we can bound the work 
extractable from a resource state via \GTO{}.
We consider the battery $W$ separately from the system $S$ of interest.
We assume that $W$ and $S$ occupy a product state.
(This assumption is unnecessary
if we focus on average work.)
Let $\initialBatt$ and $\finalBatt$
denote the battery's initial and final states.

For all $\alpha$,
\begin{align}
  F_\alpha(\initial   \otimes   \initialBatt,   \gibbsSBatt)
  \geq F_\alpha(\final   \otimes   \finalBatt,   \gibbsSBatt).
\end{align}
Since
$F_\alpha(   \initial   \otimes   \initialBatt,   \gibbsSBatt) 
= F_\alpha(\initial,   \gibbsS) +
F_\alpha   \left(\initialBatt,   \gibbsBatt   \right)$,
\begin{align}
     F_\alpha   \left( \finalBatt,   \gibbsBatt   \right)
     -   F_\alpha   \left( \initialBatt,   \gibbsBatt   \right)
     \leq
     F_\alpha(\initial,   \gibbsS) -  F_\alpha(\final,   \gibbsS).
     \label{eq:work-alpha-free-energyMAIN}
\end{align}
The left-hand side of Ineq.~\eqref{eq:work-alpha-free-energyMAIN} 
represents the work extractable during one implementation of $\initial   \to   \final$.
Hence the right-hand side 
bounds the  work extractable during the transition.

Consider extracting work from many copies of $\initial$
(i.e., extracting work from $\initial^{ \otimes n}$) in each of many trials.
Consider the average-over-trials extracted work, defined as
$\Tr(  \workf [ \finalBatt  -  \initialBatt ] )$.
The average-over-trials work extracted per copy of $\initial$ is
$\frac{1}{n}   \Tr(  \workf [ \finalBatt  -  \initialBatt ] )$.
This average work per copy has a high probability of lying close to
the change in the expectation value of the system's work function,
$\frac{1}{n} \Tr( \workf [\finalBattÊ -Ê \initialBatt] )
\approx    \Tr(  \workf [ \final  -  \initial ] )$, if $n$ is large.

Averaging over the left-hand side of Ineq.~\eqref{eq:work-alpha-free-energyMAIN} 
yields the average work $\delta \langle W \rangle$ extracted per instance of the transformation.
The average over the right-hand side
approaches the change in $F_1$ [Eq.~\eqref{eq:freeenergy}]:
\begin{align}
\delta \<  W  \>
  \leq   \delta \< H_\sys  \>_{\initial}   -   T  \,  \delta S(\initial)   
       +   \sum_{j = 1}^c   \mu_j   \,   \delta  \<Q_{j_\sys}\>.
\end{align}
This bound is achievable with a reference system, as shown in 
\cite{aberg2014catalytic,korzekwa2015extraction}.

We have focused on the extraction of work defined by $\workf$.
One can extract, instead, an individual charge $Q_j$.
The second laws do not restrict single-charge extraction.
But extracting much of one charge $Q_j$
precludes the extraction of much of another charge, $Q_k$.
 In Supplementary Note~\ref{supp-section:SI_RT}, we discuss the tradeoffs amongst
 the extraction of different charges $Q_j$.


%
%
%
\section{Discussion}
\label{section:Discussion}

We have derived, via multiple routes, the form of the thermal state
of a system that has noncommuting conserved charges.
First, we regarded the system as part of a vast composite
that occupied an approximate microcanonical state.
Tracing out the environment yields a reduced state
that lies, on average, close to a thermal state of the expected form.
This microcanonical argument, 
with canonical typicality,
suggests that the \GGS{} is the fixed point of typical dynamics.
Defining a resource theory,
we showed that the \GGS{} is the only completely passive state
and is the only state preserved by free operations.
These physical derivations
buttress Jaynes's information-theoretic derivation
from the Principle of Maximum Entropy.

Our derivations also establish tools
applicable to quantum noncommutation in thermodynamics.
In the microcanonical argument,
we introduced an approximate microcanonical state $\Omega$.
This $\Omega$ resembles the microcanonical ensemble
associated with a fixed energy, a fixed particle number, etc.\@
but accommodates noncommuting charges.
Our complete-passivity argument relies on a little-explored resource theory for thermodynamics,
in which free unitaries conserve noncommuting charges.

We expect that the equilibrium behaviors predicted here may be observed in experiments.
Quantum gases have recently demonstrated 
equilibrium-like predictions about integrable quantum systems~\cite{Rigol07,Langen15}.

%
From a conceptual perspective, our work shows that notions 
previously considered relevant only to commuting charges---for example, microcanonicals
subspace---extend to noncommuting charges.
This work opens fully quantum thermodynamics
to analysis with familiar, but suitably adapted, technical tools.


\paragraph{Data availability:}
Data sharing is not applicable to this article, as no datasets were generated or analysed during this study.

\paragraph{Acknowledgments:}
We thank David Jennings, Tim Langen, Elliott Lieb (who pointed us to~\cite{Ogata11}), Matteo Lostaglio, Shelly Moram, Joseph M. Renes, and Terry Rudolph for interesting conversations. We thank the authors of~\cite{Imperial15,teambristol} for their community spirit.  
Much of this paper was developed at ``Beyond i.i.d. in Information Theory
2015,'' hosted by BIRS.  AW's work was supported by the EU (STREP ``RAQUEL''), the ERC
(AdG ``IRQUAT''), the Spanish MINECO (grant FIS2013-40627-P) with the support of FEDER
funds, as well as by the Generalitat de Catalunya CIRIT, project 2014-SGR-966.  JO is
supported by an EPSRC Established Career Fellowship, the Royal Society, and FQXi.  NYH was
supported by an IQIM Fellowship and NSF grant PHY-0803371. The Institute for Quantum
Information and Matter (IQIM) is an NSF Physics Frontiers Center supported by the Gordon
and Betty Moore Foundation.  PhF acknowledges support from the European Research Council
(ERC) via grant No. 258932, from the Swiss National Science Foundation through the
National Centre of Competence in Research ``Quantum Science and Technology'' (QSIT), and
by the European Commission via the project ``RAQUEL.'' This work was partially supported
by the COST Action MP1209. 

\paragraph{Author contributions:}
All authors contributed equally to the 
science and to the writing of the present paper.

\paragraph{Competing financial interests:} The authors declare no competing financial interests.

%
%
%
%

\vspace{4em} 

 \section*{Supplementary information}\vspace{1em}








%
%



The microcanonical, dynamical, and resource-theory arguments are detailed below.

%
%
%
%
\section{Microcanonical derivation of the \GGS{}'s form}
  \label{section:SI_Micro}

Upon describing the set-up, we will define an approximate microcanonical subspace $\mathcal{M}$.
Normalizing the projector onto $\mathcal{M}$ yields
an approximate microcanonical state $\Omega$.
Tracing out most of the system from $\Omega$ 
leads, on average, to a state close to the \GGSlong{} $\gamma_{ \mathbf{v} }$.
Finally, we derive conditions under which $\mathcal{M}$ exists.

%
%
%
%
\paragraph{Set-up:} 
Consider a system $\mathcal{S}$ associated with 
a Hilbert space $\mathcal{H}$ of dimension $d := \dim (\mathcal{H})$.
Let $H \equiv Q_0$ denote the Hamiltonian.
We call observables denoted by 
$Q_1, \ldots, Q_c$ ``charges.''
Without loss of generality, we assume that the $Q_j$'s form a linearly independent set.
The $Q_j$'s do not necessarily commute with each other.
They commute with the Hamiltonian
if they satisfy a conservation law,
\begin{align}
   [H, Q_j]  =  0
   \; \:  \forall j = 1, \ldots, c.
\end{align}
This conservation is relevant to dynamical evolution,
during which the \GGS\ may arise as the equilibrium state.
However, our microcanonical derivation does not rely on conservation.

%
%
%
%
\paragraph{Bath, blocks, and approximations to charges:}
Consider many copies $n$ of the system $\mathcal{S}$. 
Following Ogata~\cite{Ogata11}, we consider 
an average $\tilde{Q}_j$, over the $n$ copies, of each charge $Q_j$ (Fig.~\ref{maintext-fig:OgataProofSetup} of the main text):
\begin{align}
   \tilde{Q}_j  :=
   \frac{1}{n}  \sum_{ \ell = 0 }^{ n - 1 }
   \id^{ \otimes \ell }  \otimes  Q_j  \otimes  \id^{ \otimes (n - 1 - \ell) }.
\end{align}


In the large-$n$ limit, the averages $\tilde{Q}_j$
are approximated by observables $\tilde{Y}_j$ that commute~\cite[Theorem 1.1]{Ogata11}:
\begin{align}
   & \| \tilde{Q}_j  -  \tilde{Y}_j  \|_\infty  
   \leq  \epsilon_{\mathrm{O}}(n)  \to  0,
   \text{ and } \\
   & [ \tilde{Y}_j,  \tilde{Y}_k ]  =  0
   \; \; \forall j, k  = 0, \ldots, c.
\end{align}
The $\tilde{Y}_j$'s are defined on $\mathcal{H}^{ \otimes n}$, 
$\| \cdot \|_\infty$ denotes the operator norm, and $\epsilon_{\mathrm{O}}(n)$ denotes a
  function that approaches zero as $n\to\infty$.

Consider $m$ blocks of $n$ copies of $\mathcal{S}$,
i.e., $N = nm$ copies of $\mathcal{S}$.
We can view one copy as the system of interest and 
$N - 1$ copies as a bath.
Consider the average, over $N$ copies, of a charge $Q_j$:
\begin{align}   
\label{eq:Ogata1}
   \bar{Q}_j :=
   \frac{1}{N}  \sum_{ \ell = 0}^{ N - 1}  
      \id^{\otimes \ell }  \otimes  Q_j  \otimes  \id^{ \otimes (N - 1 - \ell) }.
\end{align}
This $\bar{Q}_j$ equals also
the average, over $m$ blocks, of the block average $\tilde{Q}_j$:
\begin{align}
\label{eq:Ogata2}
   \bar{Q}_j   = 
   \frac{1}{m}  \sum_{ \lambda = 0}^{m - 1}
   \id^{ \otimes \lambda n }  \otimes  \tilde{Q}_j  \otimes  \id^{\otimes [ N - n (\lambda + 1) ] }.
\end{align}

Let us construct observables $\bar{Y}_j$ that approximate the $\bar{Q}_j$'s and that commute:
$[ \bar{Y}_j,  \bar{Y}_k ]  =  0$, and
$\| \bar{Q}_j  -  \bar{Y}_j  \|_\infty  \leq  \epsilon$
for all $m$.
Since $\tilde{Y}_j$ approximates the $\tilde{Q}_j$ in Eq.~\eqref{eq:Ogata2},
we may take
\begin{align}
   \bar{Y}_j  =
   \frac{1}{m}  \sum_{ \lambda = 0}^{m - 1}
   \id^{ \otimes \lambda n }  \otimes  \tilde{Y}_j  \otimes  \id^{\otimes [ N - n (\lambda + 1) ] }.
\end{align}

%
%
%
%
%
%
%
%

\paragraph{Approximate microcanonical subspace:}
Recall the textbook derivation
of the form of the thermal state
of a system that exchanges commuting charges with a bath.
The composite system's state occupies a microcanonical subspace.
In every state in the subspace, every whole-system charge, including the energy, 
has a well-defined value.
Charges that fail to commute 
might not have well-defined values simultaneously.
But, if $N$ is large, the $\bar{Q}_j$'s nearly commute;
they can nearly have well-defined values simultaneously.
This approximation motivates our definition
of an approximate microcanonical subspace $\mathcal{M}$.
If the composite system occupies any state in $\mathcal{M}$,
one has a high probability of being able to predict
the outcome of a measurement of any $\bar{Q}_j$.

%
%
%
%
\begin{definition}
\label{def:approx-microcanonical-subspace}
For $\eta,\eta',\epsilon, \delta, \delta' > 0$, an
\emph{$(\epsilon, \eta, \eta', \delta, \delta')$-approximate microcanonical (a.m.c.) subspace}
$\mathcal{M}$ of $\mathcal{H}^{\otimes N}$
associated with observables $Q_j$ and
with approximate expectation values $v_j$
consists of the states $\omega$ for which the
probability distribution over the possible outcomes of 
a measurement of any $\bar{Q}_j$ peaks sharply about $v_j$. 
More precisely, we denote by $\Pi_j^\eta$ the projector onto 
the direct sum of the eigensubspaces of $\bar Q_j$ 
associated with the eigenvalues in the interval $[v_j-\eta\Sigma(Q_j),v_j+\eta\Sigma(Q_j)]$.
Here, $\Sigma(Q) = \lambda_{\max}(Q)-\lambda_{\min}(Q)$ is the
spectral diameter of an observable $Q$.
$\mathcal{M}$ must satisfy the following conditions:
\begin{enumerate}
   
   \item  
   \label{item:Def1A}
   Let $\omega$ denote any state, defined on $\mathcal{H}^{\otimes N}$,
   whose support lies in $\mathcal{M}$. 
   A measurement of any $\bar{Q}_j$ is likely to
   yield a value near $v_j$:
   \begin{align}
      \label{eq:Ogata3}
      {\rm supp} (\omega)  \subset  \mathcal{M}
      \quad  \Rightarrow  \quad
      \Tr ( \omega  \Pi_j^\eta )  \geq  1 - \delta  
      \; \: \forall j.
   \end{align}
   
   \item
   \label{item:Def1B}
   Conversely, consider any state $\omega$, defined on $\mathcal{H}^{\otimes N}$,
   whose measurement statistics peak sharply. 
   Most of the state's probability weight lies in $\mathcal{M}$:
   \begin{align}
      \label{eq:Ogata3B}
      \Tr ( \omega  \Pi_j^{\eta'} )  \geq  1 - \delta'  \; \:  \forall j
      \quad  \Rightarrow \quad
      \Tr ( \omega  P)  \geq  1 - \epsilon,
   \end{align}
   wherein $P$ denotes the projector onto $\mathcal{M}$.
   
\end{enumerate}

\end{definition}

This definition merits two comments.
First, $\mathcal{M}$ is the trivial (zero) subspace 
if the $v_j$'s are inconsistent, i.e., if no state 
$\rho$ satisfies $\Tr ( \rho \, Q_j )  =  v_j  \; \;  \forall j$.
Second, specifying $(\eta,\eta',\epsilon, \delta, \delta')$
does not specify a unique subspace.
The inequalities enable multiple approximate microcanonical subspaces to satisfy 
Definition~\ref{def:approx-microcanonical-subspace}.
The definition ensures, however, that any two such subspaces
overlap substantially.

%
%
%
%
\medskip\noindent
\textbf{The approximate microcanonical subspace leads to the \GGS{}:}
Let us show that Definition~\ref{def:approx-microcanonical-subspace} exhibits the property
desired of a microcanonical state:
The reduced state of each subsystem is close to the \GGS{}.

We denoted by $P$ the projector onto the approximate microcanonical subspace $\mathcal{M}$.
Normalizing the projector yields the approximate microcanonical state
$\Omega  :=  \frac{1}{\Tr (P)  }P$.
Tracing out all subsystems but the $\ell^{\rm th}$ yields
$\Omega_\ell  :=  \Tr_{0, \ldots, \ell - 1, \ell + 1, \ldots, N-1} ( \Omega )$.

We quantify the discrepancy between $\Omega_\ell$ and the \GGS{}
with the relative entropy:
\begin{align}
   \label{eq:RelEnt}
   D ( \Omega_\ell \| \gamma_{ \mathbf{v} } )
   :=  - S( \Omega_\ell )
        -  \Tr \Big(  \Omega_\ell    \log  ( \gamma_{ \mathbf{v} } )  \Big).
\end{align}
wherein $S ( \Omega_\ell )  :=  - \Tr \Big( \Omega_\ell  \log  ( \Omega_\ell )  \Big)$
is the von Neumann entropy.
The relative entropy is lower-bounded by the trace norm,
which quantifies quantum states' distinguishability~\cite{HiaiOT81}:
\begin{align}
  \label{eq:PinskerIneq-2}
  D(\Omega_\ell\Vert\gamma_{\mathbf v})
  \geq \frac12 \norm{\Omega_\ell-\gamma_{\mathbf v}}_1^2.
\end{align}

%

\begin{theorem}
  \label{thm:microcan-implies-GGS}
  Let $\mathcal{M}$ denote an $(\epsilon,\eta,\eta',\delta,\delta')$-approximate
  microcanonical subspace of $\mathcal{H}^{\otimes N}$
  associated with the $Q_j$'s and the $v_j$'s,
  for $N \geq [2   \norm{Q_j}_\infty^2/(\eta^2) ]\log(2/\delta')$.
  The average, over the $N$ subsystems, 
  of the relative entropy between 
  each subsystem's reduced state $\Omega_\ell$ and the \GGS{} is small:
  \begin{align}
     \label{eq:AvgRelEnt}
      \frac{1}{N} \sum_{ \ell = 0}^{ N - 1} D ( \Omega_\ell \| \gamma_{ \mathbf{v} } )
       \leq \theta + \theta'.
  \end{align}
  This $\theta = \mathrm{(const.)} / \sqrt{N}$
  is proportional to a constant dependent on $\epsilon$, on the $v_j$'s, and on $d$.
  This  $\theta'   =   (c+1)   \mathrm{(const.)}   
  (\eta   +   2 \delta   \cdot   \max_j \{ \norm{Q_j}_\infty   \})$
  is proportional to a constant dependent on the $v_j$'s.
\end{theorem}

\begin{proof}
  We will bound each term in the definition~\eqref{eq:RelEnt} of the relative entropy $D$.
  The von Neumann-entropy term $S( \Omega_\ell )$,
  we bound with Schumacher's theorem for typical subspaces.
  The cross term is bounded, by the definition of 
  the approximate microcanonical subspace $\mathcal{M}$,
  in terms of the small parameters that quantify the approximation.

  First, we lower-bound the dimensionality of $\mathcal{M}$ in terms of
  $\epsilon,\eta,\eta',\delta$,  and  $\delta'$.  
  Imagine measuring some $\bar{Q}_j$ 
  of the composite-system state $\gamma_{\mathbf v}^{\otimes N}$.
  This is equivalent to 
  measuring each subsystem's $Q_j$,
  then averaging the outcomes.
  Each $Q_j$ measurement would yield a random outcome
  $X^j_\ell   \in   [\lambda_{\min}({Q_j}),\lambda_{\max}({Q_j})]$, 
  for $\ell=0,\ldots,N-1$.  
  The average of these $Q_j$-measurement outcomes
  is tightly concentrated around $v_j$, by Hoeffding's Inequality~\cite{Hoeffding63}:
  \begin{align}
      1   -   \Tr \,\bigl(   \gamma_{\mathbf v}^{\otimes N}    \Pi_j^\eta   \bigr) 
      &= \operatorname{Pr}\left\{ \Bigl\lvert 
            \frac1N   \sum_{\ell = 0}^{N - 1}    X^j_\ell - v_j
            \rvert > \eta\Sigma(Q_j) \right\}  \\
      &\leq 2 \exp\left( - 2\eta^2 N \right) \\
      &\leq \delta',
  \end{align}
  for large enough $N$.
  From the second property in Definition~\ref{def:approx-microcanonical-subspace}, 
  it follows that
  $\Tr\,\bigl(\gamma_{\mathbf v}^{\otimes N} P\bigr) \geq 1-\epsilon $.
  Hence $\mathcal{M}$ is a high-probability subspace of
  $\gamma_{\mathbf v}^{\otimes N}$.  

  By Schumacher's Theorem, or by the stronger~\cite[Theorem I.19]{PhdWinter1999}, 
  \begin{align}
    S(\Omega) 
    = \log  \Big( \dim (P) \Big) 
    &\geq N S\bigl(\gamma_{\mathbf v}\bigr)
    - (\mathrm{const.})   \sqrt{N}   \\
    &= N S\bigl(\gamma_{\mathbf v}\bigr) - N\theta,
      \label{eq:lower-bound-approx-microcan-subspace-dim}
  \end{align}
  wherein $\theta  :=  (\mathrm{const.})  /  \sqrt N$.
  The constant depends on $\epsilon$,
  $d$, and the charge values $v_j$.
  The entropy's subadditivity implies that
  $S(\Omega) \leq    \sum_{\ell = 0}^{N - 1}   S(\Omega_\ell)$. 
  Combining this inequality with Ineq.~\eqref{eq:lower-bound-approx-microcan-subspace-dim} yields
  \begin{align}
    S\bigl(\gamma_{\mathbf v}\bigr) - \theta
    \leq \frac1N    \sum_{\ell = 0}^{N - 1}    S(\Omega_\ell).
    \label{eq:lower-bound-entropy-of-average-Omegal}
  \end{align}
  
  The support of $\Omega$ lies within $\mathcal{M}$:
  $\operatorname{supp} ( \Omega )  \subset  \mathcal{M}$.   Hence 
  $\Tr(\Omega   \,   \Pi^\eta_j)   =   1   \geq   1-\delta$ for all $j$. 
  Let  $\bar\Omega  :=   \frac1N    \sum_{\ell = 0}^{N - 1}   \Omega_\ell$.
  We will bound the many-copy average
  \begin{align} 
     w_j  := \Tr(  Q_j      \,  \bar \Omega)  
     & =  \label{eq:wj1}
            \frac1N    \sum_{\ell = 0}^{N - 1}    \Tr( \Omega_\ell   \,   Q_j)  \\
     & =   \label{eq:wj2}
             \Tr( \Omega  \, \bar Q_j).
  \end{align}
  Let us bound this trace from both sides.
  Representing $\bar Q_j = \sum_{q} q\,\Pi^q_j$
  in its eigendecomposition,
  we upper-bound the following average:
\begin{widetext}
  \begin{align}
     \Tr(\Omega  \,   \bar Q_j)
     & =   \sum_q q    \Tr  \left(\Omega  \,  \Pi^q_j  \right)  \\
     &  \leq [   v_j+\eta\Sigma(Q_j)   ]   
          \Tr  \left(\Omega  \,  \Pi^\eta_j  \right)
           + \norm{Q_j}_\infty   \Tr  \Big(   \Omega       \left[ \id-\Pi^\eta_j \right]  \Big)  \\
     &  \label{eq:Bound1}
         \leq v_j + \norm{Q_j}_\infty (\eta+\delta).
  \end{align} 
  We complement this upper bound with a lower bound:
  \begin{align}
     \Tr(\Omega    \,   \bar Q_j)
     &  \geq   [  v_j-\eta\Sigma(Q_j)  ]     
                    \Tr   \left(\Omega  \,  \Pi^\eta_j   \right) -
                  \norm{Q_j}_\infty   \Tr   \Big( \Omega    \left[ \id-\Pi^\eta_j \right]   \Big)   \\
     &  \label{eq:Bound2}
        \geq   [v_j-\eta\Sigma(Q_j)  ]  (1-\delta) - \norm{Q_j}_\infty\delta.
  \end{align}
\end{widetext}
  Inequalities~\eqref{eq:Bound1} and~\eqref{eq:Bound2} show that
  the whole-system average $w_j$ is close to the single-copy average $v_j$:
  \begin{align}
     \xi_j   \label{eq:xijDefn}
     := \abs{w_j - v_j}    
     &  =   \abs{\Tr   (\Omega     \,    \bar Q_j   ) - v_j}   \\
     & \label{eq:xijBound}
         \leq  (\eta + 2\delta) \norm{Q_j}_\infty .
  \end{align}
  %

  Let us bound the average relative entropy. By definition,
  \begin{align}
     \label{eq:DExpn}
    \frac1N   \sum_{\ell = 0}^{N - 1} 
         D\,(\Omega_\ell\Vert\gamma_{\mathbf v})
    = -\frac1N   \sum_{\ell = 0}^{N - 1}   \left[ S(\Omega_\ell) 
        +     \Tr  \Big(  \Omega_\ell   \log  ( \gamma_{\mathbf v} )   \Big)   \right].
  \end{align}
  Let us focus on the second term.
  First, we substitute in the form of $\gamma_{\mathbf{v}}$ 
  from Eq.~\eqref{maintext-eq:GGS} of the main text.
  Next, we substitute in for $w_j$, using Eq.~\eqref{eq:wj1}.
  Third, we substitute in $\xi_j$, using Eq.~\eqref{eq:xijDefn}.
  Fourth, we invoke the definition of $S( \gamma_{ \mathbf{v} } )$,
  which we bound with Ineq.~\eqref{eq:lower-bound-entropy-of-average-Omegal}:
  \begin{align}
    -\frac1N   \sum_{\ell = 0}^{N - 1}
        & \Tr  \,   \Big(   \Omega_\ell   \log  ( \gamma_{\mathbf v} )  \Big) \\
    & =  \frac1N   \sum_{\ell = 0}^{N - 1}
        \Big[ \log (Z)   
        +  \sum_{j = 0}^c   \mu_j   \Tr   (\Omega_\ell  \,   Q_j)  \Big]  \\
    & = \log Z +   \sum_{j = 0}^c   \mu_j w_j  \\
    & \leq \log Z +   \sum_{j = 0}^c   \mu_j v_j 
        +   \sum_{j = 0}^c   \abs{\mu_j}\xi_j   \\
    &= S(\gamma_{\mathbf v}) 
          +  \sum_{j = 0}^c   \abs{\mu_j}\xi_j  \\
    &  \leq  \frac{1}{N}   \sum_{ \ell  =  0}^{N - 1}   S ( \Omega_\ell )  
          +  \theta   +   \sum_{j = 0}^c   \abs{\mu_j}\xi_j.
  \end{align}
  Combining this inequality with Eq.~\eqref{eq:DExpn} yields
  \begin{align}
    \hspace{1em}&\hspace{-1em}%
    \frac1N    \sum_{\ell = 0}^{N - 1} 
                   D\,(   \Omega_\ell\Vert\gamma_{\mathbf v}   )
    \leq   \theta   +   \sum_{j = 0}^c   \abs{\mu_j}\xi_j \\
    &  \leq  \theta  +    (c+1)  \,  \left( \max_j  | \mu_j | \right)  
                                         \left( \max_j \xi_j \right)    \\
    &  \label{eq:LateDBound}
        \leq  \theta  +    (c+1)  \,  \left( \max_j  | \mu_j | \right) 
        \left[ (\eta  +  2\delta) \cdot  
                \max_j  \left\{   \| Q_j \|_\infty \right\}  \right].
  \end{align}
  The final inequality follows from Ineq.~\eqref{eq:xijBound}.
  Since the $v_j$'s determine the $\mu_j$-values,
  $(c+1) \left( \max_j  | \mu_j | \right)$ is a constant determined by the $v_j$'s.
  The final term in Ineq.~\eqref{eq:LateDBound}, therefore, is upper-bounded by 
  $\theta'  =   (c+1) \mathrm{(const.)}   
  (\eta   +   2 \delta)   \cdot    \max_j   \left\{ \norm{Q_j}_\infty  \right\}$.
\end{proof}

\paragraph{Existence of an approximate microcanonical subspace:}
Definition~\ref{def:approx-microcanonical-subspace} does not reveal
under what conditions  an approximate microcanonical subspace $\mathcal{M}$ exists.
We will show that an $\mathcal{M}$ exists
for $\epsilon, \eta, \eta', \delta, \delta'$ that can approach zero simultaneously,
for  sufficiently large $N$.
First, we prove the existence of a microcanonical subspace for commuting observables.
Applying this lemma to the $\tilde{Y}_j$'s
shows that $\mathcal{M}$ exists for noncommuting observables.

%
%
%
%
\begin{lemma}
\label{lemma:MicroSubspaceCom}
  Consider a Hilbert space $\mathcal{K}$ with commuting observables
  $X_j$, $j=0,\ldots,c$. For all $\epsilon, \eta, \delta > 0$
  and for sufficiently large $m$, there exists an
  $\left(\epsilon, \eta, \eta'{=}\eta, \delta, \delta'{=}\frac{\epsilon}{c+1}\right)$-approximate
  microcanonical subspace $\mathcal{M}$ of $\mathcal{K}^{\otimes m}$ 
  associated with the observables ${X}_j$ 
  and with the approximate expectation values $v_j$.
\end{lemma}

\begin{proof}
Recall that 
\begin{equation} 
  \bar{X}_j = \frac{1}{m} \sum_{\lambda=0}^{m-1}
   \id^{\otimes \lambda} \otimes    X_{ j }    \otimes \id^{\otimes (m-1-\lambda)}
\end{equation}
is the average of $X_j$ over the $m$ subsystems.
Denote by
\begin{equation}
  \Xi_j^\eta   := \bigl\{ v_j-\eta \leq \bar{X}_j \leq v_j+\eta \bigr\}
\end{equation}
the projector onto the direct sum of the $\bar X_j$ eigenspaces 
associated with the eigenvalues in $[v_j-\eta,v_j+\eta]$.
Consider the subspace $\mathcal{M}_{\rm com}^\eta$
projected onto by all the 
$X_j$'s.
The projector onto $\mathcal{M}_{\rm com}^\eta$ is
\begin{align}
   \label{align:UpsilonProd}
   P_{\rm com}  :=  \Xi_0^\eta   \,  \Xi_1^\eta  \cdots  \Xi_c^\eta.
\end{align}


Denote by $\omega$ any state whose support lies in   $\mathcal{M}_{\rm com}^\eta$.
Let us show that $\omega$ satisfies the inequality in~\eqref{eq:Ogata3}.
By the definition of $P_\mathrm{com}$, 
${\rm supp}( \omega )   \subset   {\rm supp} ( \Xi_j^\eta )$.
Hence  $\Tr \left( \omega   \Xi_j^\eta   \right)   =   1   \geq 1-\delta$.

Let us verify the second condition in Definition~\ref{def:approx-microcanonical-subspace}.
Consider any eigenvalue $\bar{y}_j$ of $\bar{Y}_j$, for each $j$.
Consider the joint eigensubspace, shared by the $\bar{Y}_j$'s,
associated with any eigenvalue $\bar{y}_1$ of $\bar{Y}_1$,
with any eigenvalue $\bar{y}_2$ of $\bar{Y}_2$, etc.
Denote the projector onto this eigensubspace of $\mathcal{H}^{\otimes N}$ by
$\mathcal{P}_{\bar y_1, \cdots, \bar y_c}$.

Let $\delta'   =   \frac{\epsilon}{c+1}$.
Let   $\omega$ denote any state, defined on  $\mathcal{H}^{\otimes N}$, for which
$\Tr  \left(   \omega   \,   \Xi_j^\eta   \right) \geq 1-\delta'$, 
for all $j=0,\ldots,c$.
The left-hand side of the second inequality in~\eqref{eq:Ogata3B} reads, 
$\Tr \left(   \omega   P_\mathrm{com}   \right)$.
We insert the resolution of identity
$ \sum_{\bar y_0,  \ldots, \bar y_c}    \mathcal{P}_{\bar y_0  \ldots  \bar y_c}$
into the trace.
The property $\mathcal{P}^2  =  \mathcal{P}$
of any projector $\mathcal{P}$
enables us to square each projector.
Because $[\mathcal{P}_{\bar y_0  \ldots   \bar y_c},   P_\mathrm{com}]=0$,
\begin{align}
  \Tr   \left(  \omega   P_\mathrm{com}\right)
  & = \Tr   \left(   \sum_{\bar y_0,  \ldots, \bar y_c}    \mathcal{P}_{\bar y_0  \ldots  \bar y_c} 
         \omega    \mathcal{P}_{\bar y_0   \ldots   \bar y_c} P_\mathrm{com} \right)   \\
  & =: \Tr \left(   \omega'    P_\mathrm{com}   \right),
\end{align}
wherein $\omega'  :=  \sum_{\bar y_0,  \ldots, \bar y_c}    \mathcal{P}_{\bar y_0  \ldots  \bar y_c} 
         \omega    \mathcal{P}_{\bar y_0   \ldots   \bar y_c}$ 
is $\omega$ pinched with the complete set 
$\{   \mathcal{P}_{\bar y_0\bar y_1 \ldots \bar y_c}\}$
of projectors~\cite{Hayashi02}.  
By this definition of $\omega'$,
$\Tr \left(   \omega'    \,    \Xi_j^\eta   \right)
= \Tr \left(   \omega   \,    \Xi_j^\eta   \right)
\geq   1-\delta'$, and
$[\omega',   \Xi_j^\eta]=0$.
For all $j$,  therefore,
\begin{align}
   \omega'    \,  \Xi_j^\eta 
   = \omega' - \omega'\left(\id-\Xi_j^\eta   \right) 
   =: \omega' -   \Delta_j,
\end{align}
wherein
\begin{align}
   \Tr ( \Delta_j  )
   = \Tr \left(   \omega'   \left[   \id-\Xi_j^\eta   \right]   \right) 
   \leq \delta'.
\end{align}
Hence
\begin{align}
   \Tr  \left(   \omega'   P_{\rm com}   \right)
   &   =   \Tr \left(   \omega'   \,   \Xi_0^\eta      \,    \Xi_1^\eta
            \cdots   \Xi_c^\eta   \right)  \\
   &  \geq \Tr \left(   \left[  \omega'   -   \Delta_0\right]  
                               \Xi_1^\eta   \cdots\Xi_c^\eta\right)   \\
   %
   %
   & \geq \Tr \left(   \omega'   \,   \Xi_1^\eta\cdots\Xi_c^\eta\right)  - \delta'   \\
   & \geq \Tr \left(   \omega'   \right) - (c+1)\delta'   \\
   & = 1 - (c+1)\delta'
     =   1  -  \epsilon.
\end{align}
As $\omega$ satisfies~\eqref{eq:Ogata3B},
$\mathcal{M}_\text{com}^\eta$ is an
$(\epsilon, \eta, \eta'{=}\eta, \delta, \delta'{=}\frac{\epsilon}{c+1})$-approximate 
microcanonical subspace.
\end{proof}

Lemma~\ref{lemma:MicroSubspaceCom} proves the existence of an approximate microcanonical subspace 
$\mathcal{M}_{\rm com}^\eta$ for the $\tilde{Y}_j$'s defined on $\mathcal{K} = \mathcal{H}^{\otimes n}$
and for sufficiently large $n$.
In the subsequent discussion, we denote
by $\Upsilon_j^\eta$ the projector onto the direct sum of the $\bar Y_j$ eigenspaces 
associated with the eigenvalues in $[v_j-\eta\Sigma(\tilde{Y}_j),v_j+\eta\Sigma(\tilde{Y}_j)]$.
Passing from $\tilde{Y}_j$ to $\tilde{Q}_j$ to $Q_j$,
we now prove that the same $\mathcal{M}_{\rm com}^\eta$ is an
approximate microcanonical subspace for the $Q_j$'s.

%
%
%
%
\begin{theorem}
  \label{thm:mocrocononocol}
  Under the above assumptions, 
  for every $\epsilon > (c+1)\delta' > 0$, $\eta > \eta' > 0$,
  $\delta > 0$, and all sufficiently large $N$, there exists an
  $(\epsilon, \eta, \eta', \delta, \delta')$-approximate
  microcanonical subspace $\mathcal{M}$ of $\mathcal{H}^{\otimes N}$ 
  associated with the observables $Q_j$ 
  and with the approximate expectation values $v_j$.  
\end{theorem}

\begin{proof}
Let $\hat{\eta} = (\eta+\eta')/2$.
For a constant $C_{\rm AP} > 0$ to be determined later,
let $n$ be such that $\epsilon_{\rm O} = \epsilon_{\mathrm{O}}(n)$
from Ogata's result~\cite[Theorem~1.1]{Ogata11} is small enough so that
$\eta > \hat{\eta}+C_{\rm AP}\epsilon_{\rm O}^{1/3}$
and $\eta' < \hat{\eta}-C_{\rm AP}\epsilon_{\rm O}^{1/3}$,
as well as such that 
$\hat{\delta} = \delta - C_{\rm AP}\epsilon_{\rm O}^{1/3} > 0$
and such that 
$\hat{\delta}' = \delta' + C_{\rm AP}\epsilon_{\rm O}^{1/3} \leq \frac{\epsilon}{c+1}$.

Choose $m$ in Lemma~\ref{lemma:MicroSubspaceCom} large enough
such that an $(\epsilon,\hat{\eta},\hat{\eta}'{=}\hat{\eta},\hat{\delta},\hat{\delta}')$-approximate
microcanonical subspace $\mathcal{M} := \mathcal{M}_{\rm com}$ associated with the 
commuting $\tilde{Y}_j$ exists, with approximate expectation values $v_j$.

Let $\omega$ denote a state defined on $\mathcal{H}^{\otimes N}$.
We will show that,
if measuring the $\bar{Y}_j$'s of $\omega$ yields sharply peaked statistics,
measuring the $\bar{Q}_j$'s yields sharply peaked statistics.
Later, we will prove the reverse
(that sharply peaked $\bar{Q}_j$ statistics imply 
sharply peaked $\bar{Y}_j$ statistics).

Recall from Definition~\ref{def:approx-microcanonical-subspace} that
$\Pi^\eta_j$ denotes the projector onto 
the direct sum of the $\bar Q_j$ eigenstates 
associated with the eigenvalues in $[v_j-\eta\Sigma(Q_j),  \:  v_j+\eta\Sigma(Q_j)]$.
These eigenprojectors are discontinuous functions of the observables.
Hence we look for better-behaved functions.
We will approximate the action of $\Pi^\eta_j$ by using 
\begin{align}
   \label{eq:CtsFxn}
   f_{\eta_0,\eta_1}(x)  :=  \begin{cases}
               1,  &  x  \in  [ - \eta_0,  \eta_0 ]   \\
               0,  &  | x |  >  \eta_1
               \end{cases},
\end{align}
for $\eta_1>\eta_0>0$.
The Lipschitz constant of $f$ is bounded by
$\lambda := \frac{1}{\eta_1 - \eta_0} \in \mathbb{R}$.  

The operator
$f_{\eta_0\Sigma(Q_j),\eta_1\Sigma(Q_j)}(\bar Q_j - v_j\id)$
approximates the projector $\Pi_j^{\eta_0}$.
Indeed, as a matrix, $f_{\eta_0\Sigma(Q_j),\eta_1\Sigma(Q_j)}( \bar{Q}_j - v_j \id )$ 
is sandwiched between the projector $\Pi_j^{\eta_0}$,
associated with a width-$\eta_0$ interval around $v_j$,
and a projector $\Pi_j^{\eta_1}$ 
associated with a width-$\eta_1$ interval of eigenvalues.
$f_{\eta,\eta}$ is the indicator function on the interval $[-\eta,\eta]$.
Hence $\Pi_j^\eta = f_{\eta\Sigma(Q_j),\eta\Sigma(Q_j)}(\bar Q_j - v_j\id )$.
Similarly, we can regard 
$f_{\eta_0\Sigma(Q_j),\eta_1\Sigma(Q_j)}( \bar{Y}_j - v_j \id )$ as 
sandwiched between $\Upsilon_j^{\eta_0}$ 
and $\Upsilon_j^{\eta_1}$.

Because $\bar{Q}_j$ is close to $\bar{Y}_j$, $f( \bar{Q}_j )$ is close to $f ( \bar{Y}_j )$:
Let $n$ be large enough so that, by~\cite[Theorem 1.1]{Ogata11},
$\| \bar{Q}_j - \bar{Y}_j \|_\infty \leq \epsilon_\mathrm{O}$. 
By~\cite[Theorem 4.1]{AleksandrovP10},
\begin{multline}
   \label{eq:fsClose}
   \|  f_{ \eta_0\Sigma(Q_j),  \eta_1\Sigma(Q_j) }( \bar{Y}_j - v_j\id )   
      -   f_{ \eta_0\Sigma(Q_j),  \eta_1\Sigma(Q_j) } ( \bar{Q}_j - v_j\id )  \|_\infty
\\
  \leq  \kappa_\lambda,
\end{multline}
wherein $\kappa_\lambda  =  C_{\rm AP }  \lambda  \epsilon_\mathrm{O}^{2/3}$
and $C_{\rm AP}$ denotes a universal constant.
Inequality~\eqref{eq:fsClose} holds because 
$f$ is $\lambda$-Lipschitz and bounded, 
so the H\"{o}lder norm in~\cite[Theorem 4.1]{AleksandrovP10} 
is proportional to $\lambda$.

Let us show that, 
if measuring the $\bar{Y}_j$'s of $\omega$ yields sharply peaked statistics, 
then measuring the $\bar{Q}_j$'s yields sharply peaked statistics, 
and vice versa.
First, we choose $\eta_0=\eta$, $\eta_1=\eta + \epsilon_\mathrm{O}^{1/3}$, 
and $\lambda=\epsilon_\mathrm{O}^{-1/3}$ such that 
$\kappa   :=   \kappa_\lambda
=   C_\mathrm{AP}\epsilon_\mathrm{O}^{1/3}$.
By the ``sandwiching,''
\begin{align}
   \label{eq:Direc1A}
   \Tr   \left( \omega  \Pi_j^{\eta+\epsilon_\mathrm{O}^{1/3}} \right)   
   & \geq   \Tr \left(  \omega  f_{\eta_0\Sigma(Q_j),\eta_1\Sigma(Q_j)} 
                      \left[ \bar{Q}_j   -  v_j  \id \right]  \right).
\end{align}
To bound the right-hand side, we invoke Ineq.~\eqref{eq:fsClose}:
\begin{align}
   \kappa  
   &  \geq   \|  f_{ \eta_0\Sigma(Q_j),  \eta_1\Sigma(Q_j) }( \bar{Y}_j - v_j\id )   
                     \nonumber   \\ & \qquad 
                     -   f_{ \eta_0\Sigma(Q_j),  \eta_1\Sigma(Q_j) } ( \bar{Q}_j - v_j\id )  \|_\infty  \\
   &  \geq  \Tr \Big(  f_{ \eta_0\Sigma(Q_j),  \eta_1\Sigma(Q_j) }( \bar{Y}_j - v_j\id )   
                               \nonumber   \\ & \qquad 
                     -   f_{ \eta_0\Sigma(Q_j),  \eta_1\Sigma(Q_j) } ( \bar{Q}_j - v_j\id ) \Big)   \\
   &  \geq  \Tr \Big(  \omega  \Big[  
                               f_{ \eta_0\Sigma(Q_j),  \eta_1\Sigma(Q_j) }( \bar{Y}_j - v_j\id )   
                               \nonumber   \\ & \qquad 
                     -   f_{ \eta_0\Sigma(Q_j),  \eta_1v } ( \bar{Q}_j - v_j\id ) 
                               \Big]  \Big).
\end{align}
Upon invoking the trace's linearity, we rearrange terms:
\begin{align}
   \Tr  \Big(  & \omega   f_{ \eta_0\Sigma(Q_j),  \eta_1\Sigma(Q_j) } ( \bar{Q}_j - v_j\id )   \Big)
   \\ &   \geq 
   \Tr  \Big(  \omega   f_{ \eta_0\Sigma(Q_j),  \eta_1\Sigma(Q_j) }( \bar{Y}_j - v_j\id )   \Big)
   -   \kappa  \\
   &   \geq  \label{eq:Direc1B}
         \Tr  \left(  \omega  \Upsilon^\eta_j   \right)   -   \kappa.
\end{align}
The final inequality follows from the ``sandwiching'' property of 
$f_{ \eta_0,  \eta_1 }$.
Combining Ineqs.~\eqref{eq:Direc1A} and~\eqref{eq:Direc1B} yields
a bound on fluctuations in $\bar{Q}_j$ measurement statistics
in terms of fluctuations in $\bar{Y}_j$ statistics:
\begin{align}
   \label{eq:tromegaPij-geq-tromegaUpsilonj}
   \Tr  \left( \omega  \,  \Pi_j^{\eta+\epsilon_\mathrm{O}^{1/3}}   \right)
   \geq   \Tr  \left(  \omega    \Upsilon^\eta_j   \right)   -   \kappa.
\end{align}

Now, we bound fluctuations in $\bar{Y}_j$ statistics
with fluctuations in $\bar{Q}_j$ statistics.
If    $\eta_0   =   \eta-\epsilon_\mathrm{O}^{1/3}$;    $\eta_1=\eta$;
$\lambda=\epsilon_\mathrm{O}^{-1/3}$, as before, and
$\kappa=\kappa_\lambda=C_\mathrm{AP}\epsilon_\mathrm{O}^{1/3}$, then
\begin{align}
  \Tr   \left(\omega \Upsilon_j^{\eta}\right)
   \geq   \Tr  \left( \omega  \Pi_j^{\eta-\epsilon_\mathrm{O}^{1/3}}   \right)  -  \kappa.  
 \label{eq:tromegaUpsilonj-geq-tromegaPij}
\end{align}
Using Ineqs.~\eqref{eq:tromegaPij-geq-tromegaUpsilonj} and~\eqref{eq:tromegaUpsilonj-geq-tromegaPij}, 
we can now show that $\mathcal{M} := \mathcal{M}_\mathrm{com}^{\hat{\eta}}$ is an
approximate microcanonical subspace for the observables
$Q_j$ and the approximate charge values $v_j$.
In other words, $\mathcal{M}$ is an
approximate microcanonical subspace for the observables $\tilde{Q}_j$.

First, we show that $\mathcal{M}$ satisfies the first condition in Definition~\ref{def:approx-microcanonical-subspace}.
Recall that $\mathcal{M}_\mathrm{com}^\eta$ is an
$\left( \epsilon, \eta, \eta'{=}\eta, \delta, \delta'{=}\frac{\epsilon}{c} \right)$-approximate
microcanonical subspace for the observables $\tilde{Y}_j$ with the approximate charge values $v_j$,
for all $\epsilon,\eta,\delta>0$ and for large enough $m$ (Lemma~\ref{lemma:MicroSubspaceCom}).
Recall that $N=nm$.
Choose $\delta   =   \hat\delta-\kappa>0$.
Let $\omega$ denote any state, defined on $\mathcal{H}^{\otimes N}$, 
whose support lies in
$\mathcal{M}  =  \mathcal{M}_\mathrm{com}^{\eta}$. 
Let   $\hat\eta = \eta+\epsilon_\mathrm{O}^{1/3}$. 
By the definitions of $\omega$ and $\mathcal{M}$, 
$\Tr    \left(  \omega \Upsilon_j^{\eta}   \right)   =   1   \geq 1-\delta$.
By Ineq.~\eqref{eq:tromegaPij-geq-tromegaUpsilonj}, therefore,
\begin{align}
   \Tr   \left(   \omega   \Pi_j^{\hat\eta}   \right) 
   \geq \Tr  \left(   \omega\Upsilon_j^{\eta}   \right)  -   \kappa 
   \geq    1   -   \delta-\kappa 
   =   1 - \hat\delta.
\end{align}
Hence $\mathcal{M}$ satisfies Condition~\ref{item:Def1A} in Definition~\ref{def:approx-microcanonical-subspace}.

To show that $\mathcal{M}$ satisfies Condition~\ref{item:Def1B}, let
$\hat\eta'=\eta - \epsilon_\mathrm{O}^{1/3}$, and let
$\hat\delta'
=   \delta'-   \kappa   
=   \frac{\epsilon}{c}   -   C_{\mathrm{AP}}\epsilon_\mathrm{O}^{1/3} 
>0$. 
Let   $\omega$ in $\mathcal{H}^{\otimes N}$ satisfy
$\Tr   \left(   \omega \Pi_j^{\hat\eta'}   \right)   \geq 1-\hat\delta'$ for all $j$. 
By Ineq.~\eqref{eq:tromegaUpsilonj-geq-tromegaPij},
\begin{align}
   \Tr \left(   \omega\Upsilon_j^{\eta}   \right)   
   \geq 1 - \hat\delta'-\kappa
   =1-\delta'.
\end{align}
By Condition~\ref{item:Def1B} in the definition of $\mathcal{M}_\mathrm{com}^\eta$,
therefore, at least fraction $1 - \epsilon$ of the probability weight of $\omega$ 
lies in $\mathcal{M}_\mathrm{com}^\eta   =   \mathcal{M}$:
$\Tr   \left(   \omega P_\mathrm{com}   \right)   \geq 1   -   \epsilon$. 
As $\mathcal{M}$ satisfies Condition~\ref{item:Def1B}, 
$\mathcal{M}$ is an $(\epsilon,\hat\eta,\hat\eta',\hat\delta,   \hat\delta')$-approximate microcanonical subspace.
\end{proof}

This derivation confirms physically the information-theoretic 
maximum-entropy derivation. By ``physically,'' we mean, 
``involving the microcanonical form of a composite system's state
and from the tracing out of an environment.''
The noncommutation of the charges $Q_j$ required us to define 
an approximate microcanonical subspace $\mathcal{M}$.
The proof of the subspace's existence, under appropriate conditions,
crowns the derivation.

The physical principle underlying this derivation
is, roughly, the Correspondence Principle.
The $Q_j$'s of one copy of the system $\mathcal{S}$ fail to commute with each other.
This noncommutation constitutes quantum mechanical behavior.
In the many-copy limit, however,
averages $\bar{Q}_j$ of the $Q_j$'s 
are approximated by commuting $\bar{Y}_j$'s, whose
existence was proved by Ogata~\cite{Ogata11}.
In the many-copy limit,
the noncommuting (quantum) problem reduces approximately
to the commuting (classical) problem.

We stress that the approximate microcanonical subspace $\mathcal{M}$
corresponds to a set of observables $Q_j$ and a set of values $v_j$.
Consider the subspace $\mathcal{M}'$ associated with
a subset of the $Q_j$'s and their $v_j$'s.
This $\mathcal{M}'$ differs from $\mathcal{M}$.
Indeed, $\mathcal{M}'$ typically has a greater dimensionality than $\mathcal{M}$,
because fewer equations constrain it.
Furthermore, consider a linear combination 
$Q' = \sum_{j=0}^c \mu_j Q_j$. 
The average $\bar{Q'}$ of $N$ copies of $Q'$
equals $\sum_{j=0}^c \mu_j \bar{Q}_j$.
The approximate microcanonical subspace $\mathcal{M}$ of the whole
set of $Q_j$'s has the property that 
all states that lie mostly on it 
have sharply defined values near 
$v' = \sum_{j=0}^c \mu_j v_j$.
Generally, however, our $\mathcal{M}$ is not 
an approximate microcanonical subspace for $Q'$, or a selection
of $Q'$, $Q''$, etc., unless these primed operators span the same set of observables as the $Q_j$'s.

%
%
%
%
\section{Dynamical considerations}
\label{section:SI_Dynamics}

Inequality~\eqref{maintext-eq:typical-2} of the main text is derived as follows:
Let us focus on $\|  \rho_\ell  -  \gamma_{ \mathbf{v} }  \|_1$.
Adding and subtracting $\Omega_\ell$ to the argument,
then invoking the Triangle Inequality, yields
\begin{align}
   \| \rho_\ell  -  \gamma_{ \mathbf{v} } \|_1
   \leq   \|  \rho_\ell  -  \Omega_\ell \|_1  +  \|  \Omega_\ell  -  \gamma_{ \mathbf{v} }  \|_1.
\end{align}
We average over copies $\ell$ 
and average (via $\langle . \rangle$) over pure whole-system states $\ket{ \psi }$.
The first term on the right-hand side is bounded 
in Ineq.~\eqref{maintext-eq:typical} of the main text:
\begin{align}
   \label{eq:Typical1}
   \biggl\langle    \frac1N    \sum_{\ell=0}^{N-1} 
         \| \rho_\ell   -    \gamma_{\mathbf{v}}  \|_1    \biggr\rangle
   \leq   
          \frac{d}{ \sqrt{D_M} }
          +   \left\langle   \frac1N    \sum_{\ell=0}^{N-1}
           \|  \Omega_\ell  -  \gamma_{ \mathbf{v} }  \|_1   \right\rangle.
\end{align}
To bound the final term, we invoke Pinsker's Inequality [Ineq.~\eqref{eq:PinskerIneq-2}],
$\|  \Omega_\ell  -  \gamma_{ \mathbf{v} }  \|_1
\leq   \sqrt{ 2 D( \Omega_\ell  ||  \gamma_{ \mathbf{v} } ) }$.
Averaging over $\ell$ and over states $\ket{\psi}$ yields
\begin{align}
   \left\langle   \frac1N    \sum_{\ell=0}^{N-1}
      \|  \Omega_\ell  -  \gamma_{ \mathbf{v} }  \|_1
      \right\rangle
   & \leq   \left\langle   \frac1N    \sum_{\ell=0}^{N-1}
            \sqrt{ 2 D( \Omega_\ell  ||  \gamma_{ \mathbf{v} } ) }
            \right\rangle   \\
   &  \leq  \left\langle   
               \sqrt{ \frac2N    \sum_{\ell=0}^{N-1}
               D( \Omega_\ell  ||  \gamma_{ \mathbf{v} } ) }
               \right\rangle,
\end{align}
wherein $D$ denotes the relative entropy.
The second inequality follows from the square-root's concavity.
Let us double each side of Ineq.~\eqref{eq:AvgRelEnt},
then take the square-root:
\begin{align}
   \sqrt{ \frac{2}{N} \sum_{ \ell = 0}^{ N - 1} 
    D ( \Omega_\ell \| \gamma_{ \mathbf{v} } ) }
    \leq  \sqrt{ 2 (\theta + \theta') }.
\end{align}
Combining the foregoing two inequalities,
and substituting into Ineq.~\eqref{eq:Typical1}, yields Ineq.~\eqref{maintext-eq:typical-2} of the main text.

%
%
%
%
\section{Derivation from complete passivity and resource theory}
  \label{section:SI_RT}

An alternative derivation of 
the thermal state's form relies on {complete passivity}.
One cannot extract work from any number of copies of 
the thermal state via any energy-preserving unitary~\cite{PuszW78,Lenard78}.  
We adapt this argument to noncommuting conserved charges.
The {\GGSlong} is shown to be the completely passive ``free'' state 
in a  thermodynamic resource theory.

%

{Resource theories} are models, developed in quantum information theory, for scarcity.
Using a resource theory, one can calculate the value attributable to a quantum state
by an agent limited to performing only certain operations,
called ``free operations.''
The first resource theory described pure bipartite entanglement~\cite{HorodeckiHHH09}.
Entanglement theory concerns how one can manipulate entanglement,
if able to perform only local operations and classical communications.
The entanglement theory's success led to 
resource theories for asymmetry~\cite{BartlettRS07}, for stabilizer codes in quantum computation~\cite{VeitchMGE14}, for coherence~\cite{WinterY15}, for quantum Shannon theory~\cite{DevetakHW05}, and for thermodynamics, amongst other settings.

Resource-theoretic models for heat exchanges were constructed recently~\cite{janzing_thermodynamic_2000,FundLimits2}.
The free operations, called ``thermal operations,''
conserve energy.
How to extend the theory to other conserved quantities was noted in~\cite{FundLimits2}.
The commuting-observables version of the theory was defined and analyzed in~\cite{YungerHalpernR14,YungerHalpern14},
which posed questions about modeling noncommuting observables.
We extend the resource theory 
to model thermodynamic exchanges of noncommuting observables.
The free operations that define this theory,
we term ``\GTOlong'' (\GTO).
This resource theory is related to that in~\cite{Imperial15}.
We supplement earlier approaches with 
a work payoff function,
as well as with a reference frame associated with a non-Abelian group.

This section is organized as follows.
First, we introduce three subsystems and define work.
Next, we define \GTO.
The \GTO{} resource theory leads to the \GGS{} via two routes:
\begin{enumerate}
   \item 
   The \GGS{} is completely passive:
   The agent cannot extract work from
   any number of copies of $\gamma_{ \mathbf{v}}$.

   \item 
   The \GGS{} is the state preserved by \GTO{}, the operations that require no work.
\end{enumerate}
The latter condition leads to ``second laws'' for thermodynamics 
that involves noncommuting conserved charges.
The second laws imply the maximum amount of work extractable
from a transformation between states.

%
%
%
%
\paragraph{Subsystems:}
To specify a physical system in this resource theory, one specifies
a Hilbert space, a density operator, a Hamiltonian,
and operators that represent the system's charges.
To specify the subsystem $S$ of interest, for example,
one specifies a Hilbert space $\mathcal{H}$;
a density operator $\rho_\sys$; a Hamiltonian $H_\sys$;
and charges $Q_{1_\sys}, \ldots, Q_{c_\sys}$.

Consider the group $G$ formed from elements of the form 
$e^{i  \boldsymbol{\mu}   \cdot   \mathbf{Q}}$.
  Each $Q_j$ can be viewed as a generator.  $G$ is non-Abelian if the $Q_j$'s fail to
commute with each other.  Following~\cite{kitaev2014super}, we assume that $G$ is a
compact Lie group.
The compactness assumption is satisfied if the system's Hilbert space is 
finite-dimensional.
(We model the reference frame's Hilbert space as infinite-dimensional for convenience.
Finite-size references can implement the desired protocols with arbitrary
fidelity~\cite{kitaev2014super}.)

We consider three systems, apart from $S$:
First, $R$ denotes a reservoir of free states.
The resource theory is nontrivial, we prove, if and only if
the free states have the \GGS{}'s form.
Second, a battery $W$ stores work.
$W$ doubles as a non-Abelian reference frame.
Third, any other ancilla is denoted by $A$.

The Hamiltonian
$H_{\rm tot} := H_\sys   +   H_\bath   +   H_{ \batt }+H_\anc$
governs the whole system.
The $j^{\rm th}$ whole-system charge has the form
$Q_{j_{\rm tot}} := Q_{j_\sys}  +  Q_{j_\bath}  +  Q_{j_{ \batt }}  +  Q_{j_\anc}$.  
Let us introduce each subsystem individually.

\paragraph{Battery:}
We define work by modeling the system that stores the
work. In general, the mathematical expression for thermodynamic work
depends on which physical degrees of freedom a system has.
A textbook example concerns a gas,
subject to a pressure $p$,
whose volume increases by an amount $dV$.
The gas performs an amount $dW  =  p \, dV$ of work.
If a force $F$ stretches a polymer through a displacement $dx$,
$dW  =  - F \, dx$.
If a material's magnetization decreases by an amount $dM$
in the presence of a strength-$B$ magnetic field,
$dW = B\,dM$.
%


We model the ability to convert,
into a standard form of work, 
a variation in some physical quantity. 
The model consists of an observable called a ``payoff function.''  
The payoff function is defined as
\begin{align}
  \workf := \sum_{j=0}^c\mu_j Q_{j}\ .
  \label{eq:SuppInf-workfunc}
\end{align}
We generally regard the payoff function as an observable of the battery's.
We can also consider the $\workf$ of the system of interest. 
If the system whose $\workf$ we refer to is not obvious from context, 
we will use a subscript.
For example, $\workf_\batt$ denotes the battery's work function. 

One might assume that the battery exchanges only finite amounts of charges.
Under this assumption,
a realistically sized battery can implement the desired protocols with perfect fidelity~\cite{kitaev2014super}.

%
%
%
%
\paragraph{Work:}
We define as {average extracted work} $W$ 
the difference in expectation value of the payoff function $\workf$:
\begin{align}
  \label{eq:Work}
  W := \Tr   \left(\rho_{ \batt }' \workf\right) - \Tr   \left(\rho_{ \batt } \workf\right)\, .
\end{align}
The battery's initial and final states
are denoted by $\rho_{ \batt }$ and $\rho_{ \batt }'$.
If the expectation value increases, then $W>0$, and work
has been extracted from the system of interest.  
Otherwise, work has been expended.

We focus on the average work extracted in the asymptotic limit:
We consider processing many copies of the system,
then averaging over copies.
Alternatively, one could focus on 
one instance of the transformation.
The deterministic or maximal guaranteed work would quantify the protocol's efficiency
better than the average work would~\cite{dahlsten2011inadequacy,del2011thermodynamic,FundLimits2,aaberg2013truly}.

%
%
%
%
\paragraph{Reference frame:}
Reference frames have appeared in 
the thermodynamic resource theory for heat exchanges~\cite{BrandaoHORS13,aberg2014catalytic,korzekwa2015extraction}.
We introduce a non-Abelian reference frame
into the thermodynamic resource theory for noncommuting conserved charges.
Our agent's reference frame
carries a representation of the $G$ associated with the charges~\cite{kitaev2014super,BRS-refframe-review}.

The reference frame expands the set of allowed operations
from a possibly trivial set. 
A superselection rule restricts the free operations, as detailed below.
Every free unitary $U$ conserves (commutes with) each charge. 
The system charges $Q_{j_\sys}$ might not commute with each other.
In the worst case, 
the $Q_{j_\sys}$'s share no multidimensional eigensubspace.
The only unitary that conserves all such $Q_{j_\sys}$'s 
is trivial: $U  \propto  \id$.

A reference frame ``frees up'' dynamics,
enabling the system to evolve nontrivially.
A free unitary can fail to commute with a $Q_{j_\sys}$
while preserving $Q_{j_{\rm tot}}$.
This dynamics transfers charges between the system and the reference frame.

Our agent's reference frame doubles as the battery.
The reference frame and battery are combined for simplicity, 
to reduce the number of subsystems under consideration.

%
%
%
%
\paragraph{Ancillas:}
The agent could manipulate extra subsystems, called ``ancillas.''
A list $(\rho_\anc,   H_\anc,   Q_{1_\anc},   \ldots,   Q_{c_\anc})$
specifies each ancilla $A$.
Any ancillas evolve cyclically under free operations.
That is, \GTO{} preserve the ancillas' states, $\rho_\anc$.
If \GTO{} evolved ancillas acyclically, the agent could ``cheat,''
extracting work by degrading an ancilla~\cite{brandao2013second}.

Example ancillas include catalysts.
A {catalyst} facilitates a transformation 
that could not occur for free in the catalyst's absence~\cite{brandao2013second}.
Suppose that a state 
$S  =  (\rho_\sys,  H_\sys,  Q_{1_\sys},  \ldots,  Q_{c_\sys})$ cannot transform into a state 
$\tilde{S}  =  (\tilde{\rho}_\sys,  \tilde{H}_\sys,  \tilde{Q}_{1_\sys},  \ldots,  \tilde{Q}_{c_\sys})$ 
by free operations:
$S  \not\mapsto \tilde{S}$.
Some state $X  =  (\rho_\cat,  H_\cat,  Q_{1_\cat},  \ldots,  Q_{c_\cat})$ might enable 
$S  \otimes  X  \mapsto  \tilde{S}  \otimes  X$
to occur for free.
Such a facilitated transformation is called a ``catalytic operation.''

%
%
%
%
\paragraph{\GTOlong{}:}
\GTO{} are the resource theory's free operations.
\GTO{} model exchanges of heat and of charges that might not commute with each other.
\begin{definition}
\label{definition:GTO}
Every \GTOlongSing{} (\GTO{}) consists of the following three steps.
Every sequence of three such steps forms a \GTO{}:
\begin{enumerate}
\item 
\label{free}
Any number of free states $(\rho_\bath, H_\bath, Q_{1_\bath},  \ldots,  Q_{c_\bath})$ can be added.

\item 
         \label{carry-out-U}
         Any unitary $U$ that satisfies the following conditions
         can be implemented on the whole system:
	\begin{enumerate} 
	\item $U$ preserves energy: $[U,H_{\rm tot}]=0$. \label{econ}
	
	\item $U$ preserves every total charge: 
	         $[U,  Q_{j_{\rm tot}}]=0  \;  \: \forall j = 1,  \ldots, c$. \label{qcon}
	         		
	\item Any ancillas return to their original states:
	         $\Tr_{\backslash A}  (U\rho_\mathrm{tot} U^\dagger)   =   \rho_\anc$.
	\end{enumerate}
	
\item Any subsystem can be discarded (traced out). \label{trace}

\end{enumerate}
\end{definition}

Conditions~\ref{econ} and~\ref{qcon} ensure that 
the energy and the charges are conserved.
The allowed operations are $G$-invariant,
or symmetric with respect to the non-Abelian group $G$.
Conditions~\ref{econ} and~\ref{qcon} do not significantly restrict the allowed operations,
if the agent uses a reference frame. 
Suppose that the agent wishes to implement, on $S$,
some unitary $U$ that fails to commute with some $Q_{j_\sys}$. 
$U$ can be mapped to a whole-system unitary $\tilde{U}$
that conserves $Q_{j_{\rm tot}}$.
The noncommutation represents the transfer of charges to the battery,
associated with work.


The construction of $\tilde{U}$ from $U$
is described in~\cite{kitaev2014super}.
(We focus on the subset of free operations analyzed in~\cite{kitaev2014super}.)
Let $g, \phi  \in  G$ denote any elements of the symmetry group.
Let $T$ denote any subsystem (e.g., $T = S, W$).
Let $V_{\rm T}(g)$ denote a representation,
defined on the Hilbert space of system $T$, of $g$.
Let $\ket{ \phi }_{\rm T}$ denote a state of $S$
that transforms as the left regular representation of $G$:
$V_{\rm T}(g)  \ket{\phi }_{\rm T}  =  \ket{ g \phi }_{\rm T}$.
$U$ can be implemented on the system $S$ of interest
by the global unitary
\begin{align}
   \tilde{U}  :=  \int  d \phi  \:
   \ketbra{\phi}{\phi}_\batt    \otimes
   [ V_\sys ( \phi )  \,  U  \,  V_\sys^{-1} (\phi) ].
  \label{eq:kitaev-unitary-construction}
\end{align}

The construction~\eqref{eq:kitaev-unitary-construction} does not increase 
the reference frame's entropy 
if the reference is initialized to $\ket{\phi=1}_\batt$.  
This nonincrease keeps the extracted work ``clean''~\cite{Skrzypczyk13,aaberg-singleshot,BrandaoHNOW14}. 
No entropy is ``hidden'' in the reference frame $W$.
$W$ allows us to implement the unitary $U$, 
providing or storing the charges 
consumed or outputted by the system of interest.

%
%
%
%
\subsection{A zeroth law of thermodynamics: Complete passivity of the \GGSlong{}}

Which states $\rho_\bath$ should the resource-theory agent access for free?
The free states are the only states from which
work cannot be extracted via free operations.
We will ignore $S$ in this section,
treating the reservoir $R$ as the system of interest.

%
%
%
%
\paragraph{Free states in the resource theory for heat exchanges:}
Our argument about noncommuting charges will mirror 
the argument about extracting work 
when only the energy is conserved.
Consider the thermodynamic resource theory for energy conservation.
Let $H_\bath$ denote the Hamiltonian of $R$. 
The free state $\rho_\bath$ has the form
$\rho_\bath = e^{-\beta H_\bath}/Z$~\cite{BrandaoHNOW14,YungerHalpernR14}.  
This form follows from the canonical ensemble's completely passivity 
and from the nonexistence of any other completely passive state.
Complete passivity was introduced in~\cite{PuszW78,Lenard78}.

\begin{definition}[Passivity and complete passivity]
Let $\rho$ denote a state governed by a Hamiltonian $H$.
$\rho$ is passive with respect to $H$
if no free unitary $U$ 
can lower the energy expectation value of $\rho$:
\begin{align}
   \label{eq:Passive}
   \not\exists  \, U  \:  :  \:
   \Tr \left(  U \rho U^\dag \, H  \right)
   <  \Tr  \left( \rho H \right).
\end{align}
That is, work cannot be extracted from $\rho$
by any free unitary.
If work cannot be extracted from any number $n$
of copies of $\rho$, $\rho$ is completely passive with respect to $H$:
\begin{align}
   \label{eq:ComPassive}
    \forall n = 1, 2, \ldots,  \quad
    \not\exists  \, U  \:  :  \:
   \Tr \left(  U \rho^{ \otimes n} U^\dag \, H  \right)
   <  \Tr  \left( \rho^{ \otimes n} H \right).
\end{align}
\end{definition}
\noindent A free $U$ could lower the energy expectation value
only if the energy expectation value
of a work-storage system increased.
This transfer of energy
would amount to work extraction.

Conditions under which $\rho$ is passive
have been derived~\cite{PuszW78,Lenard78}:
Let $\lbrace p_i\rbrace$ and $\lbrace E_i\rbrace$ denote the eigenvalues of $\rho$ and $H$. 
$\rho$ is passive if
\begin{enumerate}
   \item
   $[\rho,H]=0$ and 
   \item
   $E_i > E_j$ implies that $p_i \leq p_j$   for all $i,j$.
\end{enumerate}
One can check that $e^{ - \beta H_\bath } / Z$ is completely passive
with respect to $H_\bath$.

No other states are completely passive (apart from the ground state).
Suppose that the agent could access 
copies of some $\rho_0  \neq  e^{ - \beta H_\bath } / Z$.
The agent could extract work via thermal operations~\cite{brandao2013second}.
Free (worthless) states could be transformed into a (valuable) resource for free.
Such a transformation would be unphysical, 
rendering the resource theory trivial, in a sense.
(As noted in ~\cite{Lostaglio2015PRX_coherence}, if a reference frame is not allowed,
the theory might be nontrivial  in that 
creating superpositions of energy eigenstates would not be possible).


%
%
%
%
\paragraph{Free states in the resource theory of \GTOlong{}:}
We have reviewed the free states in the resource theory for heat exchanges.
Similar considerations characterize the resource theory for noncommuting charges $Q_j$.
The free states, we show, have the \GGS{}'s form.
If any other state were free,
the agent could extract work for free.

%
%
\begin{theorem} \label{zeroth}
%
There exists an $m > 0$ such that 
a \GTO{} can extract a nonzero amount of chemical work from
$(\rho_\bath)^{\otimes m}$
if and only if  
$\rho_\bath \neq e^{-\beta\,( H_\bath +\sum_j \mu_j  Q_{j_\bath} )}/Z$  
for some $\beta  \in  \mathbb{R}$.
\end{theorem}

\begin{proof}
We borrow from~\cite{pusz_passive_1978,Lenard78}
the proof that canonical-type states,
and only canonical-type states,
are completely passive.
We generalize complete passivity with respect to a Hamiltonian $H$
to complete passivity with respect to the work function $\workf$.

Every free unitary preserves every global charge. 
Hence the lowering of the expectation value
of the work function $\workf$ of a system
amounts to transferring work from the system to the battery:
\begin{align}
   \label{eq:Preserve}
   \Delta  \Tr (\workf_{ \batt }   \rho_{ \batt })   
   =   -  \Delta \Tr (\workf_\bath   \rho_\bath ).
\end{align}
Just as $e^{ - \beta H } / Z$ is completely passive with respect to $H$~\cite{pusz_passive_1978,Lenard78},
the \GGS{} is completely passive with respect to $\workf_{ \bath }$ for some $\beta$.

Conversely, if $\rho_\bath$ is not of the \GGS{} form, 
it is not completely passive with respect to $\workf_\bath$. 
Some unitary $U_{\bath^{\otimes m}}$ 
lowers the energy expectation value of $\rho_\bath^{\otimes m}$,
$\Tr    ( U_{\bath^{\otimes m}}   [\rho_\bath^{\otimes m} ]
        U_{\bath^{\otimes m}}^\dagger \workf_{\bath^{\otimes m}} )
< \Tr    ( \rho_{\bath}^{\otimes m}\workf_{\bath^{\otimes m}} )$,
for some great-enough $m$.
A joint unitary defined on $R^{\otimes m}$ and $W$ 
approximates $U_{R^{\otimes m}}$ well 
and uses the system $W$ as a reference frame
[Eq.~\eqref{eq:kitaev-unitary-construction}].  
This joint unitary conserves every global charge.  
Because the expectation value of $\workf_{\bath^{\otimes m}}$ decreases, 
chemical work is transferred to the battery.
\end{proof}

%

%
%
%
%
The \GGS{} is completely passive with respect to $\workf_\bath$
but not necessarily with respect to each charge $Q_j$.
The latter lack of passivity was viewed as problematic in~\cite{Imperial15}.
The lowering of the \GGS{}'s $\<Q_j\>$'s
creates no problems in our framework, 
because free operations cannot lower the \GGS{}'s $\<\workf\>$.
The possibility of extracting charge of a desired type $Q_j$,
rather than energy, is investigated also in~\cite{teambristol}.

For example,
let the $Q_j$'s be the components $J_j$ of the spin operator $\mathbf{J}$.
Let the $z$-axis point in the direction of $\boldsymbol{\mu}$, 
and let $\mu_z>0$:
\begin{align}
   \sum_{j = 1}^3   \mu_j  J_j
   \equiv   \mu_z  J_z.
\end{align}
The \GGS{} has the form
$\rho_\bath  =  e^{ - \beta (  H_{ \bath }  -  \mu_z J_{z_{ \bath }} ) } / Z$.
This $\rho_\bath$ shares an eigenbasis with $J_{z_{ \bath }}$.
Hence the expectation value of the battery's $J_x$ charge vanishes:
 $\Tr ( \rho_\bath   J_{x_{ \bath }} )   =   0$.
 A free unitary, defined on $R$ and $W$,
can rotate the spin operator that appears in the exponential of $\rho_\bath$.
Under this unitary, the eigenstates of $\rho_\bath$ 
become eigenstates of $J_{x_{ \bath }}$.
$\Tr ( J_x   \rho_\bath )$ becomes negative;
work appears appears to be extracted ``along the $J_x$-direction''
from $\rho_\bath$.
Hence the \GGS{} appears to lack completely passivity.
The unitary, however, extracts no chemical work: 
The decrease in
$\Tr(\rho_\bath J_{x_{ \bath }} )$ is compensated for by an increase in
$\Tr ( \rho_{ \bath } J_{z_{ \bath }} )$.

Another example concerns the charges $J_i$ and
$\rho_\bath  =  e^{ - \beta (  H_{\bath}  -  \mu_z J_{z_{ \bath }} ) } / Z$.
No amount of the charge $J_z$ can be extracted from $\rho_\bath$.
But the eigenstates of $-J_z$ are inversely populated:
The eigenstate $\ket{ z }$ 
associated with the low eigenvalue $-\frac{\hbar}{2}$ of $-J_z$ 
has the small population $e^{- \beta \hbar / 2}$.
The eigenstate $\ket{ -z }$
associated with the large eigenvalue $\frac{\hbar}{2}$ of $-J_z$
has the large population $e^{ \beta \hbar / 2}$.
Hence the charge $-J_z$ can be extracted from $\rho_\bath$.
This extractability does not prevent $\rho_\bath$ from being completely passive,
according our definition.
Only the extraction of $\mathcal{W}$ corresponds to chemical work.
The extraction of just one charge does not.

The interconvertibility of types of free energy
associated with commuting charges
was noted in~\cite{YungerHalpern14}.
Let $Q_1$ and $Q_2$ denote commuting charges, and let
$\rho_\bath   =   e^{-\beta(H_{ \bath }  -  \mu_1 Q_{1_{ \bath }} -  \mu_2 Q_{2_{ \bath }})}$.
One can extract $Q_1$ work at the expense of $Q_2$ work, 
by swapping $Q_1$ and $Q_2$ 
(if an allowed unitary implements the swap).

%
%
%
%
\subsection{\GTOlong{} preserve the \GGSlong{}.}
The \GGS{}, we have shown, is the only completely passive state.
It is also the only state preserved by \GTO{}. 
\begin{theorem} \label{gibbspreserving}
Consider the resource theory, defined by \GTO{}, associated with a fixed $\beta$.
Let each free state be specified by
$(\rho_{ \bath },   H_\bath,   Q_{1_\bath},   \ldots,   Q_{c_\bath})$,
wherein $\rho_{ \bath }  :=  e^{-\beta\,( H_\bath   -  \sum_{j = 1}^c   \mu_j   Q_{j_\bath}  )}/Z$.
Suppose that the agent has access to the battery, 
associated with the payoff function~\eqref{eq:SuppInf-workfunc}.
The agent cannot, at a cost of $\langle   \mathcal{W}   \rangle   \leq 0$, 
transform any number of copies of free states into any other state.
In particular, the agent cannot change the state's $\beta$ or $\mu_j$'s.
\end{theorem}

\begin{proof}
Drawing on Theorem~\ref{zeroth},
we prove Theorem~\ref{gibbspreserving} by contradiction.
Imagine that some free operation
could transform some number $m$ of copies
of $\nats := e^{-\beta\,(H_\bath   -   \sum_j\mu_j Q_{j_\bath})}/Z$
into some other state $\nats'$:
$\nats^{ \otimes m }  \mapsto  \nats'$.
($\nats'$ could have a different form from
  the \GGS{}'s.
  Alternatively, $\nats'$ could have the same form 
  but have different $\mu_j$'s or a different $\beta$.)
$\nats'$ is not completely passive.
Work could be extracted from 
some number $n$ of copies of $\nats'$,
by Theorem~\ref{zeroth}.
By converting copies of $\nats$ into copies of $\nats'$,
and extracting work from copies of $\nats'$,
the agent could extract work from $\nats$ for free.
But work cannot be extracted from $\nats$,
by Theorem~\ref{zeroth}.
Hence $\nats^{ \otimes m}$ must not be convertible into 
any $\nats' \neq \nats$,
for all $m = 1, 2, \ldots$.
\end{proof}

%
%
%
\paragraph{Second laws:}
Consider any resource theory defined by operations that preserve some state,
e.g., states of the form
$e^{-\beta\,( H_\bath   -  \sum_{j = 1}^c   \mu_j   Q_{j_\bath}  )}/Z$.
Consider any distance measure on states 
that is contractive under the free operations.
Every state's distance from the preserved state  $\rho_{ \bath }$
decreases monotonically under the operations. 
\GTO{} can be characterized with
any distance measure from $\rho_{ \bath }$
that is contractive under completely positive trace-preserving maps.
We focus on the R\'{e}nyi divergences,
extending the second laws developed in~\cite{brandao2013second}
for the resource theory for heat exchanges.

To avoid excessive subscripting, we alter our notation for the \GGS{}.
For any subsystem $T$,
we denote by $\gibbsParam{T}$ the \GGS{} relative 
to  the fixed $\beta$, to the fixed $\mu_j$'s,
and to the Hamiltonian $H_T$ and the charges $Q_{1_T}, \ldots, Q_{c_T}$
associated with $T$.
For example, 
$\gibbsSBatt  :=  e^{ - \beta [ (H_\sys  +  H_{ \batt }) 
+  \sum_{j = 1}^c  \mu_j  (Q_{j_\sys}  +  Q_{j_{ \batt }} ) ] }  / Z$
denotes the \GGS{} associated with the system-and-battery composite.

We define the generalized free energies 
\begin{equation}
   F_\alpha(\initial,   \gibbsS)   
   :=   k_{\rm B} T  D_\alpha(\initial   \|\gibbsS)   -   \kB T  \log (Z).
   \label{eq:genfree}
\end{equation}
The classical R{\'e}nyi divergences $D_\alpha(\initial\|\gibbsS)$ are defined as 
\begin{equation}
   D_\alpha(\initial   \|   \gibbsS)
   := \frac{\sgn(\alpha)}{\alpha-1} \log 
   \left(   \sum_k    p_k^\alpha    q_k^{1-\alpha}   \right),
   \label{eq:renyidivergence-2}
\end{equation}
wherein $p_k$ and $q_k$ denote the probabilities 
of the possible outcomes
of measurements of the work function $\workf$
associated with $\initial$ and with $\gibbsS$.
The state $\initial$ of $S$ is
compared with the \GGS{} associated with $H_\sys$ and with the $Q_{j_\sys}$'s.

The $F_\alpha$'s generalize the thermodynamic free energy.
To see how, we consider transforming $n$ copies $(\initial)^{\otimes n}$ of a state $\initial$.
Consider the asymptotic limit, similar to the thermodynamic limit,
in which $n \to \infty$.
Suppose that the agent has some arbitrarily small, nonzero probability $\varepsilon$
of failing to achieve the transformation.
$\varepsilon$ can be incorporated into any $F_\alpha$ via ``smoothing''~\cite{brandao2013second}.
The smoothed $F^\varepsilon_\alpha$ per copy of $\initial$
approaches $F_1$ in the asymptotic limit~\cite{brandao2013second}:
\begin{align}
  \lim_{n\rightarrow\infty}   \frac{1}{n}   
   F^\varepsilon_\alpha &  \Big(  
        (\initial)^{\otimes n},    ( \gibbsS )^{\otimes n}    \Big)
   =   F_1  ( \initial )   \\ 
  &  =  \< H_\sys  \>_{\initial}   -   TS(\initial)   
       +   \sum_{j = 1}^c   \mu_j   \<Q_{j_\sys}\>.
\end{align}
This expression resembles the definition
$F  :=  E - TS  +  \sum_{j = 1}^c  \mu_j  Q_j$
of a thermodynamic free energy $F$.
In terms of these generalized free energies, we formulate second laws.

%
%
\begin{proposition}
  In the presence of a heat bath 
  of inverse temperature $\beta$ and chemical potentials $\mu_j$, 
  the free energies   $F_\alpha(\initial,   \gibbsS)$ 
  decrease monotonically:
  \begin{align}
     F_{\alpha}(\initial,\gibbsS) \geq F_{\alpha}(\final,\gibbsS)
     \; \: \forall \alpha\geq 0,
  \end{align}
  wherein $\initial$ and $\final$ denote the system's initial and final states.
If    
\begin{align}
   &   [\workf_\sys,   \final]   = 0   \quad {\rm and}
   \nonumber \\ &  
   F_{\alpha}(\initial,  \gibbsS) \geq F_{\alpha}(\final,\gibbsS)
\;   \:   \forall   \alpha \geq 0,
\end{align}
some catalytic \GTO\ maps $\initial$ to $\final$.
\end{proposition}

The $F_\alpha(\initial,   \gibbsS)$'s are called ``monotones.''
Under \GTO{}, the functions cannot increase.
The transformed state approaches the \GGS{} or retains its distance.

Two remarks about extraneous systems are in order.
First, the second laws clearly govern operations
during which no work is performed on the system $S$.
But the second laws also govern work performance:
Let $SW$ denote the system-and-battery composite.
The second laws govern the transformations of $SW$.
During such transformations,
work can be transferred from $W$ to $S$.

Second, the second laws govern transformations
that change the system's Hamiltonian. 
An ancilla facilitates such transformations~\cite{HO-limitations}.
Let us model the change, via external control,
of an initial Hamiltonian $H_\sys$ into $H_\sys'$.
Let $\gibbsS$ and ${\gibbsS}'$ denote the \GGS{}s
relative to $H_\sys$ and to $H_\sys'$.
The second laws become
\begin{align}
   F_{\alpha}( \initial,   \gibbsS ) 
   \geq F_{\alpha}(  \final,   {\gibbsS}'  )
   \; \: \forall \alpha   \geq   0.
\end{align}

%
%
%
%
\paragraph{Extractable work:}
In terms of the free energies, we can bound the work 
extractable from a resource state via \GTO{}.
Unlike in the previous section,
we consider the battery $W$ separately from the system $S$ of interest.
We assume that $W$ and $S$ initially occupy a product state. 
(This assumption is reasonable for the idealised, infinite-dimensional battery
we have been considering.
As we will show, the assumption can be dropped when we focus on average work.)
Let $\initialBatt$ and $\finalBatt$
denote the battery's initial and final states.
For all $\alpha$,
\begin{align}
  F_\alpha(\initial   \otimes   \initialBatt,   \gibbsSBatt)
  \geq F_\alpha(\final   \otimes   \finalBatt,   \gibbsSBatt).
\end{align}
Since
$F_\alpha(   \initial   \otimes   \initialBatt,   \gibbsSBatt) 
= F_\alpha(\initial,   \gibbsS) +
F_\alpha   \left(\initialBatt,   \gibbsBatt   \right)$,
\begin{align}
     F_\alpha   \left( \finalBatt,   \gibbsBatt   \right)
     -   F_\alpha   \left( \initialBatt,   \gibbsBatt   \right)
      \leq 
     F_\alpha(\initial,   \gibbsS) -  F_\alpha(\final,   \gibbsS).
     \label{eq:work-alpha-free-energy}
\end{align}

If the battery states 
$\initialBatt$ and $\finalBatt$ 
are energy eigenstates,
the left-hand side of Ineq.~\eqref{eq:work-alpha-free-energy} 
represents the work extractable during one implementation of the protocol.
Hence the right-hand side  
bounds the  work extractable during the transition $\initial   \mapsto   \final$. 
This bound is a necessary condition
under which work can be extracted.

%
%

When $\alpha=1$, 
we need not assume that $W$ and $S$ occupy a product state.
The reason is that subadditivity implies
$F_1(\rho_{\rm SW},\gamma_{\rm SW})
\leq F_1(\rho_\sys,   \gamma_\sys) + F_1(\rho_{W},\gamma_\batt)$. 
$F_1$ is the relevant free energy 
if only the average work is important.


\paragraph{Quantum second laws:}
As in \cite{brandao2013second}, additional laws can be derived 
in terms of quantum R{\'e}nyi divergences~\cite{HiaiMPB2010-f-divergences,Muller-LennertDSFT2013-Renyi,WildeWY2013-strong-converse,JaksicOPP2012-entropy}. 
These laws provide extra constraints
if $\initial$ (and/or $\final$) has coherences 
relative to the $\workf_\sys$ eigenbasis.
Such coherences would prevent $\initial$ from commuting
with the work function.
Such noncommutation is a signature of truly quantum behavior.
Two quantum analogues of $F_\alpha(\initial,   \gibbsS)$ are defined as
\begin{equation}
   \qalfreesimple_\alpha(\initial,   \gibbsS)
   :=   k_{\rm B} T \frac{{\rm sgn}(\alpha)}{\alpha-1} 
        \log   \Big(  \Tr \left(\initial^\alpha (\gibbsS)^{1-\alpha} \right)   \Big)
        -   k_{\rm B} T   \log (Z)
\end{equation}
and 
\begin{align}
   \qalfree_\alpha(\initial,   \gibbsS)   & :=
   \kB   T   \frac{1}{\alpha-1}\log 
   \left(   \Tr   \left( (\gibbsS)^{\frac{1-\alpha}{2\alpha}} 
             \initial (\gibbsS)^{\frac{1-\alpha}{2\alpha}}   \right)^\alpha \right) 
   \nonumber \\  & \qquad
   -   k_{\rm B}T   \log (Z).
\end{align}
The additional second laws have the following form.

\begin{proposition} 
\GTO{} can transform $\initial$ into $\final$ only if
\begin{align}
   & \qalfree_\alpha(\initial,\gibbsS) 
   \geq \qalfree_\alpha(\final,\gibbsS)
   \quad  \forall    \alpha   \geq \frac12,
   \\ 
   & \qalfree_\alpha(\gibbsS , \initial) \geq \qalfree_\alpha(\gibbsS,   \initial)
   \quad   \forall   \alpha   \in   \left[   \frac12, 1 \right], 
   \quad {\rm and}  
   \\ 
   & \qalfreesimple_\alpha(\initial, \gibbsS) 
   \geq \qalfreesimple_\alpha(\final,\gibbsS)
   \quad  \forall  \alpha  \in  [0,  2] .
\end{align}
\end{proposition}
\noindent These laws govern transitions during which the Hamiltonian changes via an ancilla, 
as in~\cite{HO-limitations}.

%
%


%
%
%
\bibliography{NatCommsRefs.bibolamazi.bib}

\begin{thebibliography}{53}%
\makeatletter
\providecommand \@ifxundefined [1]{%
 \@ifx{#1\undefined}
}%
\providecommand \@ifnum [1]{%
 \ifnum #1\expandafter \@firstoftwo
 \else \expandafter \@secondoftwo
 \fi
}%
\providecommand \@ifx [1]{%
 \ifx #1\expandafter \@firstoftwo
 \else \expandafter \@secondoftwo
 \fi
}%
\providecommand \natexlab [1]{#1}%
\providecommand \enquote  [1]{``#1''}%
\providecommand \bibnamefont  [1]{#1}%
\providecommand \bibfnamefont [1]{#1}%
\providecommand \citenamefont [1]{#1}%
\providecommand \href@noop [0]{\@secondoftwo}%
\providecommand \href [0]{\begingroup \@sanitize@url \@href}%
\providecommand \@href[1]{\@@startlink{#1}\@@href}%
\providecommand \@@href[1]{\endgroup#1\@@endlink}%
\providecommand \@sanitize@url [0]{\catcode `\\12\catcode `\$12\catcode
  `\&12\catcode `\#12\catcode `\^12\catcode `\_12\catcode `\%12\relax}%
\providecommand \@@startlink[1]{}%
\providecommand \@@endlink[0]{}%
\providecommand \url  [0]{\begingroup\@sanitize@url \@url }%
\providecommand \@url [1]{\endgroup\@href {#1}{\urlprefix }}%
\providecommand \urlprefix  [0]{URL }%
\providecommand \Eprint [0]{\href }%
\providecommand \doibase [0]{http://dx.doi.org/}%
\providecommand \selectlanguage [0]{\@gobble}%
\providecommand \bibinfo  [0]{\@secondoftwo}%
\providecommand \bibfield  [0]{\@secondoftwo}%
\providecommand \translation [1]{[#1]}%
\providecommand \BibitemOpen [0]{}%
\providecommand \bibitemStop [0]{}%
\providecommand \bibitemNoStop [0]{.\EOS\space}%
\providecommand \EOS [0]{\spacefactor3000\relax}%
\providecommand \BibitemShut  [1]{\csname bibitem#1\endcsname}%
\let\auto@bib@innerbib\@empty
\bibitem [{\citenamefont {Gemmer}\ \emph {et~al.}(2009)\citenamefont {Gemmer},
  \citenamefont {Michel}, \citenamefont {Michel},\ and\ \citenamefont
  {Mahler}}]{gemmer2009quantum}%
  \BibitemOpen
  \bibfield  {author} {\bibinfo {author} {\bibfnamefont {J.}~\bibnamefont
  {Gemmer}}, \bibinfo {author} {\bibfnamefont {M.}~\bibnamefont {Michel}},
  \bibinfo {author} {\bibfnamefont {M.}~\bibnamefont {Michel}}, \ and\ \bibinfo
  {author} {\bibfnamefont {G.}~\bibnamefont {Mahler}},\ }\href@noop {} {\emph
  {\bibinfo {title} {Quantum thermodynamics: Emergence of thermodynamic
  behavior within composite quantum systems}}}\ (\bibinfo  {publisher}
  {Springer Verlag},\ \bibinfo {year} {2009})\BibitemShut {NoStop}%
\bibitem [{\citenamefont {Gogolin}\ and\ \citenamefont
  {Eisert}(2016)}]{Gogolin2015arXiv_review}%
  \BibitemOpen
  \bibfield  {author} {\bibinfo {author} {\bibfnamefont {C.}~\bibnamefont
  {Gogolin}}\ and\ \bibinfo {author} {\bibfnamefont {J.}~\bibnamefont
  {Eisert}},\ }\href {\doibase 10.1088/0034-4885/79/5/056001} {\bibfield
  {journal} {\bibinfo  {journal} {Reports on Progress in Physics}\ }\textbf
  {\bibinfo {volume} {79}},\ \bibinfo {pages} {056001} (\bibinfo {year}
  {2016})}\BibitemShut {NoStop}%
\bibitem [{\citenamefont {Goold}\ \emph {et~al.}(2016)\citenamefont {Goold},
  \citenamefont {Huber}, \citenamefont {Riera}, \citenamefont {del Rio},\ and\
  \citenamefont {Skrzypczyk}}]{Goold2015arXiv_review}%
  \BibitemOpen
  \bibfield  {author} {\bibinfo {author} {\bibfnamefont {J.}~\bibnamefont
  {Goold}}, \bibinfo {author} {\bibfnamefont {M.}~\bibnamefont {Huber}},
  \bibinfo {author} {\bibfnamefont {A.}~\bibnamefont {Riera}}, \bibinfo
  {author} {\bibfnamefont {L.}~\bibnamefont {del Rio}}, \ and\ \bibinfo
  {author} {\bibfnamefont {P.}~\bibnamefont {Skrzypczyk}},\ }\href {\doibase
  10.1088/1751-8113/49/14/143001} {\bibfield  {journal} {\bibinfo  {journal}
  {Journal of Physics A: Mathematical and Theoretical}\ }\textbf {\bibinfo
  {volume} {49}},\ \bibinfo {pages} {143001} (\bibinfo {year}
  {2016})}\BibitemShut {NoStop}%
\bibitem [{\citenamefont {Vinjanampathy}\ and\ \citenamefont
  {Anders}()}]{Vinjanampathy2015arXiv_review}%
  \BibitemOpen
  \bibfield  {author} {\bibinfo {author} {\bibfnamefont {S.}~\bibnamefont
  {Vinjanampathy}}\ and\ \bibinfo {author} {\bibfnamefont {J.}~\bibnamefont
  {Anders}},\ }\href@noop {} {\ }\bibinfo {note} {{Quantum Thermodynamics}.
  {Preprint} at \url{http://arxiv.org/abs/1508.06099} (2015)}\BibitemShut
  {NoStop}%
\bibitem [{\citenamefont {Goldstein}\ \emph {et~al.}(2006)\citenamefont
  {Goldstein}, \citenamefont {Lebowitz}, \citenamefont {Tumulka},\ and\
  \citenamefont {Zangh{\'\i}}}]{goldstein2006canonical}%
  \BibitemOpen
  \bibfield  {author} {\bibinfo {author} {\bibfnamefont {S.}~\bibnamefont
  {Goldstein}}, \bibinfo {author} {\bibfnamefont {J.~L.}\ \bibnamefont
  {Lebowitz}}, \bibinfo {author} {\bibfnamefont {R.}~\bibnamefont {Tumulka}}, \
  and\ \bibinfo {author} {\bibfnamefont {N.}~\bibnamefont {Zangh{\'\i}}},\
  }\href@noop {} {\bibfield  {journal} {\bibinfo  {journal} {Physical review
  letters}\ }\textbf {\bibinfo {volume} {96}},\ \bibinfo {pages} {050403}
  (\bibinfo {year} {2006})}\BibitemShut {NoStop}%
\bibitem [{\citenamefont {Gemmer}\ \emph {et~al.}(2004)\citenamefont {Gemmer},
  \citenamefont {Michel},\ and\ \citenamefont {Mahler}}]{gemmer200418}%
  \BibitemOpen
  \bibfield  {author} {\bibinfo {author} {\bibfnamefont {J.}~\bibnamefont
  {Gemmer}}, \bibinfo {author} {\bibfnamefont {M.}~\bibnamefont {Michel}}, \
  and\ \bibinfo {author} {\bibfnamefont {G.}~\bibnamefont {Mahler}},\
  }\href@noop {} {\emph {\bibinfo {title} {18 Equilibrium Properties of Model
  Systems}}}\ (\bibinfo  {publisher} {Springer},\ \bibinfo {year}
  {2004})\BibitemShut {NoStop}%
\bibitem [{\citenamefont {{Popescu}}\ \emph {et~al.}(2006)\citenamefont
  {{Popescu}}, \citenamefont {{Short}},\ and\ \citenamefont
  {{Winter}}}]{PopescuSW06}%
  \BibitemOpen
  \bibfield  {author} {\bibinfo {author} {\bibfnamefont {S.}~\bibnamefont
  {{Popescu}}}, \bibinfo {author} {\bibfnamefont {A.~J.}\ \bibnamefont
  {{Short}}}, \ and\ \bibinfo {author} {\bibfnamefont {A.}~\bibnamefont
  {{Winter}}},\ }\href@noop {} {\bibfield  {journal} {\bibinfo  {journal}
  {Nature Physics}\ }\textbf {\bibinfo {volume} {2}},\ \bibinfo {pages} {754}
  (\bibinfo {year} {2006})}\BibitemShut {NoStop}%
\bibitem [{\citenamefont {{Linden}}\ \emph {et~al.}(2009)\citenamefont
  {{Linden}}, \citenamefont {{Popescu}}, \citenamefont {{Short}},\ and\
  \citenamefont {{Winter}}}]{LindenPSW09}%
  \BibitemOpen
  \bibfield  {author} {\bibinfo {author} {\bibfnamefont {N.}~\bibnamefont
  {{Linden}}}, \bibinfo {author} {\bibfnamefont {S.}~\bibnamefont {{Popescu}}},
  \bibinfo {author} {\bibfnamefont {A.~J.}\ \bibnamefont {{Short}}}, \ and\
  \bibinfo {author} {\bibfnamefont {A.}~\bibnamefont {{Winter}}},\ }\href
  {\doibase 10.1103/PhysRevE.79.061103} {\bibfield  {journal} {\bibinfo
  {journal} {Phys. Rev. E}\ }\textbf {\bibinfo {volume} {79}},\ \bibinfo {eid}
  {061103} (\bibinfo {year} {2009})}\BibitemShut {NoStop}%
\bibitem [{\citenamefont {Fermi}\ \emph {et~al.}(1955)\citenamefont {Fermi},
  \citenamefont {Pasta},\ and\ \citenamefont {Ulam}}]{FermiPU55}%
  \BibitemOpen
  \bibfield  {author} {\bibinfo {author} {\bibfnamefont {E.}~\bibnamefont
  {Fermi}}, \bibinfo {author} {\bibfnamefont {J.}~\bibnamefont {Pasta}}, \ and\
  \bibinfo {author} {\bibfnamefont {S.}~\bibnamefont {Ulam}},\ }\href@noop {}
  {\bibfield  {journal} {\bibinfo  {journal} {Los Alamos Report LA-1940}\ }
  (\bibinfo {year} {1955})}\BibitemShut {NoStop}%
\bibitem [{\citenamefont {Kinoshita}\ \emph {et~al.}(2006)\citenamefont
  {Kinoshita}, \citenamefont {Wenger},\ and\ \citenamefont
  {Weiss}}]{KinoshitaWW06}%
  \BibitemOpen
  \bibfield  {author} {\bibinfo {author} {\bibfnamefont {T.}~\bibnamefont
  {Kinoshita}}, \bibinfo {author} {\bibfnamefont {T.}~\bibnamefont {Wenger}}, \
  and\ \bibinfo {author} {\bibfnamefont {D.~S.}\ \bibnamefont {Weiss}},\ }\href
  {\doibase 10.1038/nature04693} {\bibfield  {journal} {\bibinfo  {journal}
  {Nature}\ }\textbf {\bibinfo {volume} {440}},\ \bibinfo {pages} {900}
  (\bibinfo {year} {2006})}\BibitemShut {NoStop}%
\bibitem [{\citenamefont {Rigol}\ \emph {et~al.}(2007)\citenamefont {Rigol},
  \citenamefont {Dunjko}, \citenamefont {Yurovsky},\ and\ \citenamefont
  {Olshanii}}]{Rigol07}%
  \BibitemOpen
  \bibfield  {author} {\bibinfo {author} {\bibfnamefont {M.}~\bibnamefont
  {Rigol}}, \bibinfo {author} {\bibfnamefont {V.}~\bibnamefont {Dunjko}},
  \bibinfo {author} {\bibfnamefont {V.}~\bibnamefont {Yurovsky}}, \ and\
  \bibinfo {author} {\bibfnamefont {M.}~\bibnamefont {Olshanii}},\ }\href
  {\doibase 10.1103/PhysRevLett.98.050405} {\bibfield  {journal} {\bibinfo
  {journal} {Phys. Rev. Lett.}\ }\textbf {\bibinfo {volume} {98}},\ \bibinfo
  {pages} {050405} (\bibinfo {year} {2007})}\BibitemShut {NoStop}%
\bibitem [{\citenamefont {Polkovnikov}\ \emph {et~al.}(2011)\citenamefont
  {Polkovnikov}, \citenamefont {Sengupta}, \citenamefont {Silva},\ and\
  \citenamefont {Vengalattore}}]{PolkovnikovSSV11}%
  \BibitemOpen
  \bibfield  {author} {\bibinfo {author} {\bibfnamefont {A.}~\bibnamefont
  {Polkovnikov}}, \bibinfo {author} {\bibfnamefont {K.}~\bibnamefont
  {Sengupta}}, \bibinfo {author} {\bibfnamefont {A.}~\bibnamefont {Silva}}, \
  and\ \bibinfo {author} {\bibfnamefont {M.}~\bibnamefont {Vengalattore}},\
  }\href {\doibase 10.1103/RevModPhys.83.863} {\bibfield  {journal} {\bibinfo
  {journal} {Rev. Mod. Phys.}\ }\textbf {\bibinfo {volume} {83}},\ \bibinfo
  {pages} {863} (\bibinfo {year} {2011})}\BibitemShut {NoStop}%
\bibitem [{\citenamefont {{Langen}}\ \emph {et~al.}(2015)\citenamefont
  {{Langen}}, \citenamefont {{Erne}}, \citenamefont {{Geiger}}, \citenamefont
  {{Rauer}}, \citenamefont {{Schweigler}}, \citenamefont {{Kuhnert}},
  \citenamefont {{Rohringer}}, \citenamefont {{Mazets}}, \citenamefont
  {{Gasenzer}},\ and\ \citenamefont {{Schmiedmayer}}}]{Langen15}%
  \BibitemOpen
  \bibfield  {author} {\bibinfo {author} {\bibfnamefont {T.}~\bibnamefont
  {{Langen}}}, \bibinfo {author} {\bibfnamefont {S.}~\bibnamefont {{Erne}}},
  \bibinfo {author} {\bibfnamefont {R.}~\bibnamefont {{Geiger}}}, \bibinfo
  {author} {\bibfnamefont {B.}~\bibnamefont {{Rauer}}}, \bibinfo {author}
  {\bibfnamefont {T.}~\bibnamefont {{Schweigler}}}, \bibinfo {author}
  {\bibfnamefont {M.}~\bibnamefont {{Kuhnert}}}, \bibinfo {author}
  {\bibfnamefont {W.}~\bibnamefont {{Rohringer}}}, \bibinfo {author}
  {\bibfnamefont {I.~E.}\ \bibnamefont {{Mazets}}}, \bibinfo {author}
  {\bibfnamefont {T.}~\bibnamefont {{Gasenzer}}}, \ and\ \bibinfo {author}
  {\bibfnamefont {J.}~\bibnamefont {{Schmiedmayer}}},\ }\href {\doibase
  10.1126/science.1257026} {\bibfield  {journal} {\bibinfo  {journal}
  {Science}\ }\textbf {\bibinfo {volume} {348}},\ \bibinfo {pages} {207}
  (\bibinfo {year} {2015})}\BibitemShut {NoStop}%
\bibitem [{\citenamefont {Langen}\ \emph {et~al.}(2015)\citenamefont {Langen},
  \citenamefont {Geiger},\ and\ \citenamefont {Schmiedmayer}}]{LangenGS15}%
  \BibitemOpen
  \bibfield  {author} {\bibinfo {author} {\bibfnamefont {T.}~\bibnamefont
  {Langen}}, \bibinfo {author} {\bibfnamefont {R.}~\bibnamefont {Geiger}}, \
  and\ \bibinfo {author} {\bibfnamefont {J.}~\bibnamefont {Schmiedmayer}},\
  }\href {\doibase 10.1146/annurev-conmatphys-031214-014548} {\bibfield
  {journal} {\bibinfo  {journal} {Annual Review of Condensed Matter Physics}\
  }\textbf {\bibinfo {volume} {6}},\ \bibinfo {pages} {201} (\bibinfo {year}
  {2015})}\BibitemShut {NoStop}%
\bibitem [{\citenamefont {Janzing}\ \emph {et~al.}(2000)\citenamefont
  {Janzing}, \citenamefont {Wocjan}, \citenamefont {Zeier}, \citenamefont
  {Geiss},\ and\ \citenamefont {Beth}}]{janzing_thermodynamic_2000}%
  \BibitemOpen
  \bibfield  {author} {\bibinfo {author} {\bibfnamefont {D.}~\bibnamefont
  {Janzing}}, \bibinfo {author} {\bibfnamefont {P.}~\bibnamefont {Wocjan}},
  \bibinfo {author} {\bibfnamefont {R.}~\bibnamefont {Zeier}}, \bibinfo
  {author} {\bibfnamefont {R.}~\bibnamefont {Geiss}}, \ and\ \bibinfo {author}
  {\bibfnamefont {T.}~\bibnamefont {Beth}},\ }\href@noop {} {\bibfield
  {journal} {\bibinfo  {journal} {Int. J. Theor. Phys.}\ }\textbf {\bibinfo
  {volume} {39}},\ \bibinfo {pages} {2717} (\bibinfo {year}
  {2000})}\BibitemShut {NoStop}%
\bibitem [{\citenamefont {Brand\~ao}\ \emph {et~al.}(2013)\citenamefont
  {Brand\~ao}, \citenamefont {Horodecki}, \citenamefont {Oppenheim},
  \citenamefont {Renes},\ and\ \citenamefont {Spekkens}}]{BrandaoHORS13}%
  \BibitemOpen
  \bibfield  {author} {\bibinfo {author} {\bibfnamefont {F.~G. S.~L.}\
  \bibnamefont {Brand\~ao}}, \bibinfo {author} {\bibfnamefont {M.}~\bibnamefont
  {Horodecki}}, \bibinfo {author} {\bibfnamefont {J.}~\bibnamefont
  {Oppenheim}}, \bibinfo {author} {\bibfnamefont {J.~M.}\ \bibnamefont
  {Renes}}, \ and\ \bibinfo {author} {\bibfnamefont {R.~W.}\ \bibnamefont
  {Spekkens}},\ }\href {\doibase 10.1103/PhysRevLett.111.250404} {\bibfield
  {journal} {\bibinfo  {journal} {Physical Review Letters}\ }\textbf {\bibinfo
  {volume} {111}},\ \bibinfo {pages} {250404} (\bibinfo {year}
  {2013})}\BibitemShut {NoStop}%
\bibitem [{\citenamefont {{Brandao}}\ \emph {et~al.}(2015)\citenamefont
  {{Brandao}}, \citenamefont {{Horodecki}}, \citenamefont {{Ng}}, \citenamefont
  {{Oppenheim}},\ and\ \citenamefont {{Wehner}}}]{brandao2013second}%
  \BibitemOpen
  \bibfield  {author} {\bibinfo {author} {\bibfnamefont {F.~G. S.~L.}\
  \bibnamefont {{Brandao}}}, \bibinfo {author} {\bibfnamefont {M.}~\bibnamefont
  {{Horodecki}}}, \bibinfo {author} {\bibfnamefont {N.~H.~Y.}\ \bibnamefont
  {{Ng}}}, \bibinfo {author} {\bibfnamefont {J.}~\bibnamefont {{Oppenheim}}}, \
  and\ \bibinfo {author} {\bibfnamefont {S.}~\bibnamefont {{Wehner}}},\ }\href
  {\doibase 10.1073/pnas.1411728112} {\bibfield  {journal} {\bibinfo  {journal}
  {Proc. Natl. Acad. Sci.}\ }\textbf {\bibinfo {volume} {112}},\ \bibinfo
  {pages} {3275} (\bibinfo {year} {2015})}\BibitemShut {NoStop}%
\bibitem [{\citenamefont {Horodecki}\ and\ \citenamefont
  {Oppenheim}(2013)}]{HO-limitations}%
  \BibitemOpen
  \bibfield  {author} {\bibinfo {author} {\bibfnamefont {M.}~\bibnamefont
  {Horodecki}}\ and\ \bibinfo {author} {\bibfnamefont {J.}~\bibnamefont
  {Oppenheim}},\ }\href {\doibase 10.1038/ncomms3059} {\bibfield  {journal}
  {\bibinfo  {journal} {Nature Communications}\ }\textbf {\bibinfo {volume}
  {4}},\ \bibinfo {pages} {2059} (\bibinfo {year} {2013})}\BibitemShut
  {NoStop}%
\bibitem [{\citenamefont {Vaccaro}\ and\ \citenamefont
  {Barnett}(2011)}]{VaccaroB11}%
  \BibitemOpen
  \bibfield  {author} {\bibinfo {author} {\bibfnamefont {J.~A.}\ \bibnamefont
  {Vaccaro}}\ and\ \bibinfo {author} {\bibfnamefont {S.~M.}\ \bibnamefont
  {Barnett}},\ }\href {\doibase 10.1098/rspa.2010.0577} {\bibfield  {journal}
  {\bibinfo  {journal} {Proceedings of the Royal Society of London A:
  Mathematical, Physical and Engineering Sciences}\ }\textbf {\bibinfo {volume}
  {467}},\ \bibinfo {pages} {1770} (\bibinfo {year} {2011})}\BibitemShut
  {NoStop}%
\bibitem [{\citenamefont {Yunger~Halpern}\ and\ \citenamefont
  {Renes}(2016)}]{halpern2014beyond}%
  \BibitemOpen
  \bibfield  {author} {\bibinfo {author} {\bibfnamefont {N.}~\bibnamefont
  {Yunger~Halpern}}\ and\ \bibinfo {author} {\bibfnamefont {J.~M.}\
  \bibnamefont {Renes}},\ }\href {\doibase 10.1103/PhysRevE.93.022126}
  {\bibfield  {journal} {\bibinfo  {journal} {Phys. Rev. E}\ }\textbf {\bibinfo
  {volume} {93}},\ \bibinfo {pages} {022126} (\bibinfo {year}
  {2016})}\BibitemShut {NoStop}%
\bibitem [{\citenamefont {{Yunger Halpern}}()}]{YungerHalpern14}%
  \BibitemOpen
  \bibfield  {author} {\bibinfo {author} {\bibfnamefont {N.}~\bibnamefont
  {{Yunger Halpern}}},\ }\href@noop {} {\ }\bibinfo {note} {{Beyond heat baths
  II: Framework for generalized thermodynamic resource theories}. {Preprint} at
  \url{http://arxiv.org/abs/1409.7845} (2014)}\BibitemShut {NoStop}%
\bibitem [{\citenamefont {Weilenmann}\ \emph {et~al.}()\citenamefont
  {Weilenmann}, \citenamefont {Kr{\"{a}}mer}, \citenamefont {Faist},\ and\
  \citenamefont {Renner}}]{Weilenmann2015arXiv_axiomatic}%
  \BibitemOpen
  \bibfield  {author} {\bibinfo {author} {\bibfnamefont {M.}~\bibnamefont
  {Weilenmann}}, \bibinfo {author} {\bibfnamefont {L.}~\bibnamefont
  {Kr{\"{a}}mer}}, \bibinfo {author} {\bibfnamefont {P.}~\bibnamefont {Faist}},
  \ and\ \bibinfo {author} {\bibfnamefont {R.}~\bibnamefont {Renner}},\
  }\href@noop {} {\ }\bibinfo {note} {{Axiomatic relation between thermodynamic
  and information-theoretic entropies}. {Preprint} at
  \url{http://arxiv.org/abs/1501.06920} (2015)}\BibitemShut {NoStop}%
\bibitem [{\citenamefont {Jaynes}(1957{\natexlab{a}})}]{Jaynes57I}%
  \BibitemOpen
  \bibfield  {author} {\bibinfo {author} {\bibfnamefont {E.~T.}\ \bibnamefont
  {Jaynes}},\ }\href@noop {} {\bibfield  {journal} {\bibinfo  {journal} {Phys.
  Rev.}\ }\textbf {\bibinfo {volume} {106}},\ \bibinfo {pages} {620} (\bibinfo
  {year} {1957}{\natexlab{a}})}\BibitemShut {NoStop}%
\bibitem [{\citenamefont {Jaynes}(1957{\natexlab{b}})}]{Jaynes57II}%
  \BibitemOpen
  \bibfield  {author} {\bibinfo {author} {\bibfnamefont {E.~T.}\ \bibnamefont
  {Jaynes}},\ }\href@noop {} {\bibfield  {journal} {\bibinfo  {journal} {Phys.
  Rev.}\ }\textbf {\bibinfo {volume} {108}},\ \bibinfo {pages} {171} (\bibinfo
  {year} {1957}{\natexlab{b}})}\BibitemShut {NoStop}%
\bibitem [{\citenamefont {Pusz}\ and\ \citenamefont
  {Woronowicz}(1978)}]{pusz_passive_1978}%
  \BibitemOpen
  \bibfield  {author} {\bibinfo {author} {\bibfnamefont {W.}~\bibnamefont
  {Pusz}}\ and\ \bibinfo {author} {\bibfnamefont {S.~L.}\ \bibnamefont
  {Woronowicz}},\ }\href {\doibase 10.1007/BF01614224} {\bibfield  {journal}
  {\bibinfo  {journal} {Comm. Math. Phys.}\ }\textbf {\bibinfo {volume} {58}},\
  \bibinfo {pages} {273} (\bibinfo {year} {1978})}\BibitemShut {NoStop}%
\bibitem [{\citenamefont {Lenard}(1978)}]{Lenard78}%
  \BibitemOpen
  \bibfield  {author} {\bibinfo {author} {\bibfnamefont {A.}~\bibnamefont
  {Lenard}},\ }\href {\doibase 10.1007/BF01011769} {\bibfield  {journal}
  {\bibinfo  {journal} {J. Stat. Phys.}\ }\textbf {\bibinfo {volume} {19}},\
  \bibinfo {pages} {575} (\bibinfo {year} {1978})}\BibitemShut {NoStop}%
\bibitem [{\citenamefont {{Lostaglio}}\ \emph {et~al.}()\citenamefont
  {{Lostaglio}}, \citenamefont {{Jennings}},\ and\ \citenamefont
  {{Rudolph}}}]{Imperial15}%
  \BibitemOpen
  \bibfield  {author} {\bibinfo {author} {\bibfnamefont {M.}~\bibnamefont
  {{Lostaglio}}}, \bibinfo {author} {\bibfnamefont {D.}~\bibnamefont
  {{Jennings}}}, \ and\ \bibinfo {author} {\bibfnamefont {T.}~\bibnamefont
  {{Rudolph}}},\ }\href@noop {} {\ }\bibinfo {note} {{Thermodynamic resource
  theories, non-commutativity and maximum entropy principles}. {Preprint} at
  \url{http://arxiv.org/abs/1511.04420} (2015)}\BibitemShut {NoStop}%
\bibitem [{\citenamefont {{Guryanova}}\ \emph {et~al.}()\citenamefont
  {{Guryanova}}, \citenamefont {{Popescu}}, \citenamefont {{Short}},
  \citenamefont {{Silva}},\ and\ \citenamefont {{Skrzypczyk}}}]{teambristol}%
  \BibitemOpen
  \bibfield  {author} {\bibinfo {author} {\bibfnamefont {Y.}~\bibnamefont
  {{Guryanova}}}, \bibinfo {author} {\bibfnamefont {S.}~\bibnamefont
  {{Popescu}}}, \bibinfo {author} {\bibfnamefont {A.~J.}\ \bibnamefont
  {{Short}}}, \bibinfo {author} {\bibfnamefont {R.}~\bibnamefont {{Silva}}}, \
  and\ \bibinfo {author} {\bibfnamefont {P.}~\bibnamefont {{Skrzypczyk}}},\
  }\href@noop {} {\ }\bibinfo {note} {{Thermodynamics of quantum systems with
  multiple conserved quantities}. {Preprint} at
  \url{http://arxiv.org/abs/1512.01190} (2015)}\BibitemShut {NoStop}%
\bibitem [{\citenamefont {{Ogata}}(2013)}]{Ogata11}%
  \BibitemOpen
  \bibfield  {author} {\bibinfo {author} {\bibfnamefont {Y.}~\bibnamefont
  {{Ogata}}},\ }\href@noop {} {\bibfield  {journal} {\bibinfo  {journal}
  {Journal of Functional Analysis}\ }\textbf {\bibinfo {volume} {264}},\
  \bibinfo {pages} {2005} (\bibinfo {year} {2013})}\BibitemShut {NoStop}%
\bibitem [{\citenamefont {Hiai}\ \emph {et~al.}(1981)\citenamefont {Hiai},
  \citenamefont {Ohya},\ and\ \citenamefont {Tsukada}}]{HiaiOT81}%
  \BibitemOpen
  \bibfield  {author} {\bibinfo {author} {\bibfnamefont {F.}~\bibnamefont
  {Hiai}}, \bibinfo {author} {\bibfnamefont {M.}~\bibnamefont {Ohya}}, \ and\
  \bibinfo {author} {\bibfnamefont {M.}~\bibnamefont {Tsukada}},\ }\href@noop
  {} {\bibfield  {journal} {\bibinfo  {journal} {Pacific J. Math.}\ }\textbf
  {\bibinfo {volume} {96}},\ \bibinfo {pages} {99} (\bibinfo {year}
  {1981})}\BibitemShut {NoStop}%
\bibitem [{\citenamefont {Schumacher}(1995)}]{Schumacher95}%
  \BibitemOpen
  \bibfield  {author} {\bibinfo {author} {\bibfnamefont {B.}~\bibnamefont
  {Schumacher}},\ }\href {\doibase 10.1103/PhysRevA.51.2738} {\bibfield
  {journal} {\bibinfo  {journal} {Phys. Rev. A}\ }\textbf {\bibinfo {volume}
  {51}},\ \bibinfo {pages} {2738} (\bibinfo {year} {1995})}\BibitemShut
  {NoStop}%
\bibitem [{\citenamefont {Horodecki}\ \emph {et~al.}(2009)\citenamefont
  {Horodecki}, \citenamefont {Horodecki}, \citenamefont {Horodecki},\ and\
  \citenamefont {Horodecki}}]{horodecki_quantum_2009}%
  \BibitemOpen
  \bibfield  {author} {\bibinfo {author} {\bibfnamefont {R.}~\bibnamefont
  {Horodecki}}, \bibinfo {author} {\bibfnamefont {P.}~\bibnamefont
  {Horodecki}}, \bibinfo {author} {\bibfnamefont {M.}~\bibnamefont
  {Horodecki}}, \ and\ \bibinfo {author} {\bibfnamefont {K.}~\bibnamefont
  {Horodecki}},\ }\href@noop {} {\bibfield  {journal} {\bibinfo  {journal}
  {Rev. Mod. Phys.}\ }\textbf {\bibinfo {volume} {81}},\ \bibinfo {pages} {865}
  (\bibinfo {year} {2009})}\BibitemShut {NoStop}%
\bibitem [{\citenamefont {Aharonov}\ and\ \citenamefont
  {Susskind}(1967)}]{aharonov-susskind}%
  \BibitemOpen
  \bibfield  {author} {\bibinfo {author} {\bibfnamefont {Y.}~\bibnamefont
  {Aharonov}}\ and\ \bibinfo {author} {\bibfnamefont {L.}~\bibnamefont
  {Susskind}},\ }\href {\doibase 10.1103/PhysRev.155.1428} {\bibfield
  {journal} {\bibinfo  {journal} {Phys. Rev.}\ }\textbf {\bibinfo {volume}
  {155}},\ \bibinfo {pages} {1428} (\bibinfo {year} {1967})}\BibitemShut
  {NoStop}%
\bibitem [{\citenamefont {{Kitaev}}\ \emph {et~al.}(2004)\citenamefont
  {{Kitaev}}, \citenamefont {{Mayers}},\ and\ \citenamefont
  {{Preskill}}}]{kitaev2014super}%
  \BibitemOpen
  \bibfield  {author} {\bibinfo {author} {\bibfnamefont {A.}~\bibnamefont
  {{Kitaev}}}, \bibinfo {author} {\bibfnamefont {D.}~\bibnamefont {{Mayers}}},
  \ and\ \bibinfo {author} {\bibfnamefont {J.}~\bibnamefont {{Preskill}}},\
  }\href {\doibase 10.1103/PhysRevA.69.052326} {\bibfield  {journal} {\bibinfo
  {journal} {Phys. Rev. A}\ }\textbf {\bibinfo {volume} {69}},\ \bibinfo {eid}
  {052326} (\bibinfo {year} {2004})}\BibitemShut {NoStop}%
\bibitem [{\citenamefont {{Bartlett}}\ \emph {et~al.}(2007)\citenamefont
  {{Bartlett}}, \citenamefont {{Rudolph}},\ and\ \citenamefont
  {{Spekkens}}}]{BRS-refframe-review}%
  \BibitemOpen
  \bibfield  {author} {\bibinfo {author} {\bibfnamefont {S.~D.}\ \bibnamefont
  {{Bartlett}}}, \bibinfo {author} {\bibfnamefont {T.}~\bibnamefont
  {{Rudolph}}}, \ and\ \bibinfo {author} {\bibfnamefont {R.~W.}\ \bibnamefont
  {{Spekkens}}},\ }\href {\doibase 10.1103/RevModPhys.79.555} {\bibfield
  {journal} {\bibinfo  {journal} {Reviews of Modern Physics}\ }\textbf
  {\bibinfo {volume} {79}},\ \bibinfo {pages} {555} (\bibinfo {year}
  {2007})}\BibitemShut {NoStop}%
\bibitem [{\citenamefont {{Hiai}}\ \emph {et~al.}(2011)\citenamefont {{Hiai}},
  \citenamefont {{Mosonyi}}, \citenamefont {{Petz}},\ and\ \citenamefont
  {{B{\'e}ny}}}]{HiaiMPB2010-f-divergences}%
  \BibitemOpen
  \bibfield  {author} {\bibinfo {author} {\bibfnamefont {F.}~\bibnamefont
  {{Hiai}}}, \bibinfo {author} {\bibfnamefont {M.}~\bibnamefont {{Mosonyi}}},
  \bibinfo {author} {\bibfnamefont {D.}~\bibnamefont {{Petz}}}, \ and\ \bibinfo
  {author} {\bibfnamefont {C.}~\bibnamefont {{B{\'e}ny}}},\ }\href {\doibase
  10.1142/S0129055X11004412} {\bibfield  {journal} {\bibinfo  {journal} {Rev.
  Math. Phys.}\ }\textbf {\bibinfo {volume} {23}},\ \bibinfo {pages} {691}
  (\bibinfo {year} {2011})}\BibitemShut {NoStop}%
\bibitem [{\citenamefont {{M{\"u}ller-Lennert}}\ \emph
  {et~al.}(2013)\citenamefont {{M{\"u}ller-Lennert}}, \citenamefont {{Dupuis}},
  \citenamefont {{Szehr}}, \citenamefont {{Fehr}},\ and\ \citenamefont
  {{Tomamichel}}}]{Muller-LennertDSFT2013-Renyi}%
  \BibitemOpen
  \bibfield  {author} {\bibinfo {author} {\bibfnamefont {M.}~\bibnamefont
  {{M{\"u}ller-Lennert}}}, \bibinfo {author} {\bibfnamefont {F.}~\bibnamefont
  {{Dupuis}}}, \bibinfo {author} {\bibfnamefont {O.}~\bibnamefont {{Szehr}}},
  \bibinfo {author} {\bibfnamefont {S.}~\bibnamefont {{Fehr}}}, \ and\ \bibinfo
  {author} {\bibfnamefont {M.}~\bibnamefont {{Tomamichel}}},\ }\href@noop {}
  {\bibfield  {journal} {\bibinfo  {journal} {Journal of Mathematical Physics}\
  }\textbf {\bibinfo {volume} {54}} (\bibinfo {year} {2013})}\BibitemShut
  {NoStop}%
\bibitem [{\citenamefont {{Wilde}}\ \emph {et~al.}(2014)\citenamefont
  {{Wilde}}, \citenamefont {{Winter}},\ and\ \citenamefont
  {{Yang}}}]{WildeWY2013-strong-converse}%
  \BibitemOpen
  \bibfield  {author} {\bibinfo {author} {\bibfnamefont {M.~M.}\ \bibnamefont
  {{Wilde}}}, \bibinfo {author} {\bibfnamefont {A.}~\bibnamefont {{Winter}}}, \
  and\ \bibinfo {author} {\bibfnamefont {D.}~\bibnamefont {{Yang}}},\ }\href
  {\doibase 10.1007/s00220-014-2122-x} {\bibfield  {journal} {\bibinfo
  {journal} {Communications in Mathematical Physics}\ }\textbf {\bibinfo
  {volume} {331}},\ \bibinfo {pages} {593} (\bibinfo {year}
  {2014})}\BibitemShut {NoStop}%
\bibitem [{\citenamefont {Jaksic}\ \emph {et~al.}(2012)\citenamefont {Jaksic},
  \citenamefont {Ogata}, \citenamefont {Pautrat},\ and\ \citenamefont
  {Pillet}}]{JaksicOPP2012-entropy}%
  \BibitemOpen
  \bibfield  {author} {\bibinfo {author} {\bibfnamefont {V.}~\bibnamefont
  {Jaksic}}, \bibinfo {author} {\bibfnamefont {Y.}~\bibnamefont {Ogata}},
  \bibinfo {author} {\bibfnamefont {Y.}~\bibnamefont {Pautrat}}, \ and\
  \bibinfo {author} {\bibfnamefont {C.-A.}\ \bibnamefont {Pillet}},\ }in\ \href
  {\doibase 10.1093/acprof:oso/9780199652495.003.0004} {\emph {\bibinfo
  {booktitle} {Quantum Theory from Small to Large Scales: Lecture Notes of the
  Les Houches Summer School}}},\ Vol.~\bibinfo {volume} {95},\ \bibinfo
  {editor} {edited by\ \bibinfo {editor} {\bibfnamefont {J.}~\bibnamefont
  {Fr\"ohlich}}, \bibinfo {editor} {\bibfnamefont {S.}~\bibnamefont {Manfred}},
  \bibinfo {editor} {\bibfnamefont {M.}~\bibnamefont {Vieri}}, \bibinfo
  {editor} {\bibfnamefont {W.}~\bibnamefont {De~Roeck}}, \ and\ \bibinfo
  {editor} {\bibfnamefont {L.~F.}\ \bibnamefont {Cugliandolo}}}\ (\bibinfo
  {publisher} {Oxford University Press},\ \bibinfo {year} {2012})\BibitemShut
  {NoStop}%
\bibitem [{\citenamefont {\AA{}berg}(2014)}]{aberg2014catalytic}%
  \BibitemOpen
  \bibfield  {author} {\bibinfo {author} {\bibfnamefont {J.}~\bibnamefont
  {\AA{}berg}},\ }\href {\doibase 10.1103/PhysRevLett.113.150402} {\bibfield
  {journal} {\bibinfo  {journal} {Phys. Rev. Lett.}\ }\textbf {\bibinfo
  {volume} {113}},\ \bibinfo {eid} {150402} (\bibinfo {year}
  {2014})}\BibitemShut {NoStop}%
\bibitem [{\citenamefont {{Korzekwa}}\ \emph {et~al.}(2016)\citenamefont
  {{Korzekwa}}, \citenamefont {{Lostaglio}}, \citenamefont {{Oppenheim}},\ and\
  \citenamefont {{Jennings}}}]{korzekwa2015extraction}%
  \BibitemOpen
  \bibfield  {author} {\bibinfo {author} {\bibfnamefont {K.}~\bibnamefont
  {{Korzekwa}}}, \bibinfo {author} {\bibfnamefont {M.}~\bibnamefont
  {{Lostaglio}}}, \bibinfo {author} {\bibfnamefont {J.}~\bibnamefont
  {{Oppenheim}}}, \ and\ \bibinfo {author} {\bibfnamefont {D.}~\bibnamefont
  {{Jennings}}},\ }\href@noop {} {\bibfield  {journal} {\bibinfo  {journal}
  {New Journal of Physics}\ }\textbf {\bibinfo {volume} {18}} (\bibinfo {year}
  {2016})}\BibitemShut {NoStop}%
\bibitem [{\citenamefont {Hoeffding}(1963)}]{Hoeffding63}%
  \BibitemOpen
  \bibfield  {author} {\bibinfo {author} {\bibfnamefont {W.}~\bibnamefont
  {Hoeffding}},\ }\href {\doibase 10.1080/01621459.1963.10500830} {\bibfield
  {journal} {\bibinfo  {journal} {Journal of the American Statistical
  Association}\ }\textbf {\bibinfo {volume} {58}},\ \bibinfo {pages} {13}
  (\bibinfo {year} {1963})}\BibitemShut {NoStop}%
\bibitem [{\citenamefont {Winter}(1999)}]{PhdWinter1999}%
  \BibitemOpen
  \bibfield  {author} {\bibinfo {author} {\bibfnamefont {A.}~\bibnamefont
  {Winter}},\ }\emph {\bibinfo {title} {{Coding Theorems of Quantum Information
  Theory}}},\ \href@noop {} {Ph.D. thesis},\ \bibinfo  {school}
  {Universit{\"a}t Bielefeld} (\bibinfo {year} {1999})\BibitemShut {NoStop}%
\bibitem [{\citenamefont {{Hayashi}}(2002)}]{Hayashi02}%
  \BibitemOpen
  \bibfield  {author} {\bibinfo {author} {\bibfnamefont {M.}~\bibnamefont
  {{Hayashi}}},\ }\href {\doibase 10.1088/0305-4470/35/50/307} {\bibfield
  {journal} {\bibinfo  {journal} {Journal of Physics A Mathematical General}\
  }\textbf {\bibinfo {volume} {35}},\ \bibinfo {pages} {10759} (\bibinfo {year}
  {2002})}\BibitemShut {NoStop}%
\bibitem [{\citenamefont {Aleksandrov}\ and\ \citenamefont
  {Peller}(2010)}]{AleksandrovP10}%
  \BibitemOpen
  \bibfield  {author} {\bibinfo {author} {\bibfnamefont {A.}~\bibnamefont
  {Aleksandrov}}\ and\ \bibinfo {author} {\bibfnamefont {V.}~\bibnamefont
  {Peller}},\ }\href {\doibase http://dx.doi.org/10.1016/j.aim.2009.12.018}
  {\bibfield  {journal} {\bibinfo  {journal} {Advances in Mathematics}\
  }\textbf {\bibinfo {volume} {224}},\ \bibinfo {pages} {910 } (\bibinfo {year}
  {2010})}\BibitemShut {NoStop}%
\bibitem [{\citenamefont {{Veitch}}\ \emph {et~al.}(2014)\citenamefont
  {{Veitch}}, \citenamefont {{Hamed Mousavian}}, \citenamefont {{Gottesman}},\
  and\ \citenamefont {{Emerson}}}]{VeitchMGE14}%
  \BibitemOpen
  \bibfield  {author} {\bibinfo {author} {\bibfnamefont {V.}~\bibnamefont
  {{Veitch}}}, \bibinfo {author} {\bibfnamefont {S.~A.}\ \bibnamefont {{Hamed
  Mousavian}}}, \bibinfo {author} {\bibfnamefont {D.}~\bibnamefont
  {{Gottesman}}}, \ and\ \bibinfo {author} {\bibfnamefont {J.}~\bibnamefont
  {{Emerson}}},\ }\href {\doibase 10.1088/1367-2630/16/1/013009} {\bibfield
  {journal} {\bibinfo  {journal} {New Journal of Physics}\ }\textbf {\bibinfo
  {volume} {16}},\ \bibinfo {eid} {013009} (\bibinfo {year}
  {2014})}\BibitemShut {NoStop}%
\bibitem [{\citenamefont {Winter}\ and\ \citenamefont
  {Yang}(2016)}]{WinterY15}%
  \BibitemOpen
  \bibfield  {author} {\bibinfo {author} {\bibfnamefont {A.}~\bibnamefont
  {Winter}}\ and\ \bibinfo {author} {\bibfnamefont {D.}~\bibnamefont {Yang}},\
  }\href {\doibase 10.1103/PhysRevLett.116.120404} {\bibfield  {journal}
  {\bibinfo  {journal} {Phys. Rev. Lett.}\ }\textbf {\bibinfo {volume} {116}},\
  \bibinfo {pages} {120404} (\bibinfo {year} {2016})}\BibitemShut {NoStop}%
\bibitem [{\citenamefont {{Devetak}}\ \emph {et~al.}(2008)\citenamefont
  {{Devetak}}, \citenamefont {{Harrow}},\ and\ \citenamefont
  {{Winter}}}]{DevetakHW05}%
  \BibitemOpen
  \bibfield  {author} {\bibinfo {author} {\bibfnamefont {I.}~\bibnamefont
  {{Devetak}}}, \bibinfo {author} {\bibfnamefont {A.~W.}\ \bibnamefont
  {{Harrow}}}, \ and\ \bibinfo {author} {\bibfnamefont {A.}~\bibnamefont
  {{Winter}}},\ }\href {\doibase 10.1109/TIT.2008.928980} {\bibfield  {journal}
  {\bibinfo  {journal} {IEEE Trans. Inf. Theor.}\ }\textbf {\bibinfo {volume}
  {54}},\ \bibinfo {pages} {4587} (\bibinfo {year} {2008})}\BibitemShut
  {NoStop}%
\bibitem [{\citenamefont {Dahlsten}\ \emph {et~al.}(2011)\citenamefont
  {Dahlsten}, \citenamefont {Renner}, \citenamefont {Rieper},\ and\
  \citenamefont {Vedral}}]{workvalue}%
  \BibitemOpen
  \bibfield  {author} {\bibinfo {author} {\bibfnamefont {O.~C.~O.}\
  \bibnamefont {Dahlsten}}, \bibinfo {author} {\bibfnamefont {R.}~\bibnamefont
  {Renner}}, \bibinfo {author} {\bibfnamefont {E.}~\bibnamefont {Rieper}}, \
  and\ \bibinfo {author} {\bibfnamefont {V.}~\bibnamefont {Vedral}},\
  }\href@noop {} {\bibfield  {journal} {\bibinfo  {journal} {New Journal of
  Physics}\ }\textbf {\bibinfo {volume} {13}},\ \bibinfo {pages} {053015}
  (\bibinfo {year} {2011})}\BibitemShut {NoStop}%
\bibitem [{\citenamefont {Del~Rio}\ \emph {et~al.}(2011)\citenamefont
  {Del~Rio}, \citenamefont {{\AA}berg}, \citenamefont {Renner}, \citenamefont
  {Dahlsten},\ and\ \citenamefont {Vedral}}]{del2011thermodynamic}%
  \BibitemOpen
  \bibfield  {author} {\bibinfo {author} {\bibfnamefont {L.}~\bibnamefont
  {Del~Rio}}, \bibinfo {author} {\bibfnamefont {J.}~\bibnamefont {{\AA}berg}},
  \bibinfo {author} {\bibfnamefont {R.}~\bibnamefont {Renner}}, \bibinfo
  {author} {\bibfnamefont {O.}~\bibnamefont {Dahlsten}}, \ and\ \bibinfo
  {author} {\bibfnamefont {V.}~\bibnamefont {Vedral}},\ }\href@noop {}
  {\bibfield  {journal} {\bibinfo  {journal} {Nature}\ }\textbf {\bibinfo
  {volume} {474}},\ \bibinfo {pages} {61} (\bibinfo {year} {2011})}\BibitemShut
  {NoStop}%
\bibitem [{\citenamefont {{\AA}berg}(2013)}]{aaberg-singleshot}%
  \BibitemOpen
  \bibfield  {author} {\bibinfo {author} {\bibfnamefont {J.}~\bibnamefont
  {{\AA}berg}},\ }\href {\doibase 10.1038/ncomms2712} {\bibfield  {journal}
  {\bibinfo  {journal} {Nature Communications}\ }\textbf {\bibinfo {volume}
  {4}},\ \bibinfo {eid} {1925} (\bibinfo {year} {2013})}\BibitemShut {NoStop}%
\bibitem [{\citenamefont {{Skrzypczyk}}\ \emph {et~al.}(2014)\citenamefont
  {{Skrzypczyk}}, \citenamefont {{Short}},\ and\ \citenamefont
  {{Popescu}}}]{Skrzypczyk13}%
  \BibitemOpen
  \bibfield  {author} {\bibinfo {author} {\bibfnamefont {P.}~\bibnamefont
  {{Skrzypczyk}}}, \bibinfo {author} {\bibfnamefont {A.~J.}\ \bibnamefont
  {{Short}}}, \ and\ \bibinfo {author} {\bibfnamefont {S.}~\bibnamefont
  {{Popescu}}},\ }\href {\doibase 10.1038/ncomms5185} {\bibfield  {journal}
  {\bibinfo  {journal} {Nature Communications}\ }\textbf {\bibinfo {volume}
  {5}},\ \bibinfo {pages} {4185} (\bibinfo {year} {2014})}\BibitemShut
  {NoStop}%
\bibitem [{\citenamefont {Lostaglio}\ \emph {et~al.}(2015)\citenamefont
  {Lostaglio}, \citenamefont {Korzekwa}, \citenamefont {Jennings},\ and\
  \citenamefont {Rudolph}}]{Lostaglio2015PRX_coherence}%
  \BibitemOpen
  \bibfield  {author} {\bibinfo {author} {\bibfnamefont {M.}~\bibnamefont
  {Lostaglio}}, \bibinfo {author} {\bibfnamefont {K.}~\bibnamefont {Korzekwa}},
  \bibinfo {author} {\bibfnamefont {D.}~\bibnamefont {Jennings}}, \ and\
  \bibinfo {author} {\bibfnamefont {T.}~\bibnamefont {Rudolph}},\ }\href
  {\doibase 10.1103/PhysRevX.5.021001} {\bibfield  {journal} {\bibinfo
  {journal} {Physical Review X}\ }\textbf {\bibinfo {volume} {5}},\ \bibinfo
  {pages} {021001} (\bibinfo {year} {2015})}\BibitemShut {NoStop}%
\end{thebibliography}%

\end{document}